\documentclass[11pt,british,refpage,intoc,bibliography=totoc,index=totoc,BCOR=7.5mm,captions=tableheading]{extarticle}
\usepackage{mathptmx}
\usepackage[T1]{fontenc}
\usepackage[latin9]{inputenc}
\usepackage[a4paper]{geometry}
\usepackage{color}
\usepackage{babel}
\usepackage{algorithm}
\usepackage{algpseudocode}
\usepackage{amsmath}
\usepackage{amsthm}
\usepackage{amssymb}
\usepackage{graphicx}
\usepackage{subfig}
\usepackage[numbers]{natbib}
\usepackage[unicode=true,pdfusetitle,
 bookmarks=true,bookmarksnumbered=true,bookmarksopen=false,
 breaklinks=false,pdfborder={0 0 0},pdfborderstyle={},backref=false,colorlinks=true]
 {hyperref}
\hypersetup{
 linkcolor=black, citecolor=blue, urlcolor=black, filecolor=blue, pdfpagelayout=OneColumn, pdfnewwindow=true, pdfstartview=XYZ, plainpages=false}

\makeatletter
\numberwithin{equation}{section}
\numberwithin{figure}{section}
\theoremstyle{plain}
\newtheorem{thm}{\protect\theoremname}[section]
\theoremstyle{plain}

\newtheorem{rem}[thm]{\protect\remarkname}

\@ifundefined{date}{}{\date{}}
\usepackage{caption}
\usepackage[nottoc]{tocbibind}
\allowdisplaybreaks[4]
\usepackage{dsfont}
\DeclareMathAlphabet{\mathcal}{OMS}{cmsy}{m}{n}

\makeatother

\providecommand{\remarkname}{Remark}
\providecommand{\lemmaname}{Lemma}
\providecommand{\theoremname}{Theorem}

\begin{document}

\title{On the dynamics of the boundary vorticity\\  for incompressible viscous flows}
\author{By V. Cherepanov\thanks{Mathematical Institute, University of Oxford, Oxford OX2 6GG. Email:
\protect\href{mailto:vladislav.cherepanov@maths.ox.ac.uk}{vladislav.cherepanov@maths.ox.ac.uk}},  \ \  J. Liu\thanks{Department of Mathematics and Department of Physics, Duke University, Durham, NC 27708. Email:  
\protect\href{mailto:jian-guo.liu@duke.edu}{jian-guo.liu@duke.edu}} \ \  and \ \   Z. Qian\thanks{Mathematical Institute, University of Oxford, Oxford OX2 6GG and Oxford Suzhou Centre for Advanced Research, Suzhou, China. Email:
\protect\href{mailto:qianz@maths.ox.ac.uk}{qianz@maths.ox.ac.uk}}}

\maketitle
\begin{abstract}
The dynamical equation of the boundary vorticity has been obtained, which shows that the viscosity at a solid wall is doubled as
if the fluid became more viscous at the boundary.  For certain viscous flows the boundary vorticity can be determined via the dynamical equation up to
bounded errors for all time,
without the need of knowing the details of the main stream flows. We then validate
the dynamical equation by carrying out stochastic direct numerical 
simulations (i.e. the random vortex method for wall-bounded incompressible viscous flows) 
by two different means of updating the boundary vorticity, one using 
 mollifiers of the Biot-Savart singular integral kernel, another 
 using the dynamical equations. 

\medskip


\medskip

\end{abstract}

\section{Introduction}

When a viscous flow moves along a solid wall with large velocity, substantial molecular force takes effect among fluid particles at the boundary, and therefore vorticity is created instantly within a thin boundary layer, which in turn leads to substantial stress at the wall. The stress at the wall is indeed proportional to the vorticity created near the boundary, called the boundary vorticity for short. In many engineering applications, it is very important to understand the distribution of the stress over the boundary surface when a viscous fluid flow past a solid fluid boundary. It is important to obtain quantitative information of the stress distribution across the boundary at any instance for an unsteady viscous flow.  Information about the boundary vorticity may be gained by performing numerical computations. The finite difference method or other numerical schemes may be used for solving numerically the fluid dynamics equations or the boundary layer equations, which however require to calculate the outer layer flows as well. It is therefore not cheap to carry out numerical experiments to acquire knowledge on the distribution of the boundary vorticity in general.

In this paper we propose a different approach to the study of the boundary vorticity of an incompressible viscous fluid flow past a solid wall, motivated by the recent work on the random vortex method for wall-bounded flows (cf. \citep{QQZW2022, Qian2022})  via ordinary McKean-Vlasov type stochastic differential equations. In the random vortex methods for wall-bounded flows, the boundary stress has to be updated through iterations, and can not be assigned a priori. We instead in this work shall determine the dynamics of the boundary vorticity directly. The dynamical evolution equations for boundary vorticity for incompressible viscous fluid flows are obtained, which we believe is a new discovery. The dynamical equations of the boundary vorticity reveal several remarkable properties of incompressible viscous fluid flows at the boundary which we wish to report in this paper. It is revealed that the viscosity of the fluid flow is exactly doubled at the boundary as if the fluid became more viscous at the boundary. The dynamical equation of the boundary vorticity also demonstrates that the boundary vorticity evolves mainly linearly, in contrast to the high non-linearity of the Navier-Stokes equations. For some fluid flows, the boundary stress can be determined for all time with a bounded error, a fact which comes up a little bit surprising.

The paper is organised as the following. In Section \ref{The fluid dynamics equations for flows past a wall}, we write a formulation of the vorticity transport equation as a non-homogeneous boundary problem dependent on the boundary vorticity. The dynamical equation satisfied by the boundary vorticity is derived in Section \ref{Dynamics of the boundary vorticity}. We write the stochastic representations of the vorticity and the velocity in terms of the Taylor diffusion in Section \ref{Functional integral representations}. Using these representations, we derive and implement a numerical scheme in Section \ref{Numerical experiments} where the results of the conducted experiments are reported. 

\section{The fluid dynamics equations for flows past a wall}\label{The fluid dynamics equations for flows past a wall}

For a viscous fluid flow past a solid wall,
it is clear that the geometry of the solid wall which constrains the
fluid flow has a significant impact on the dynamics of the boundary
vortices. As a matter of fact, the dynamics of the vortex motion at
the solid wall becomes significantly complicated if the solid wall
possesses non-trivial geometry (i.e., with non-constant curvature), and therefore
the study for flows past curved surfaces will be published in a future work. In this article, we shall deal with viscous fluid flows past a flat
plate, i.e. for the case the fluid boundary has trivial geometry.

Therefore we shall consider an incompressible fluid flow constrained in
the upper half space $D=\mathbb{R}_{+}^{d}$ (where $d=2$ or $3$
in this work), the solid plate is modelled by the boundary $\partial D$ where $x_{d}=0$. Let $u=(u^{1},\ldots,u^{d})$ be the velocity
and $P$ the pressure of the fluid flow in question. Then $u(x,t)$
is a time dependent vector field in $D$. The motion of the fluid
is determined by the Navier-Stokes equations
\begin{equation}
\frac{\partial u}{\partial t}+(u\cdot\nabla)u-\nu\Delta u+\nabla P-F=0\quad\textrm{ in }D,\label{3D-Ns01}
\end{equation}
\begin{equation}
\nabla\cdot u=0\quad\textrm{ in }D,\label{3D-Ns02}
\end{equation}
together with the non-slip condition that $u(x,t)=0$ for $x\in\partial D$,
where $F=(F^{1},\ldots,F^{d})$ is the external force applied to
the fluid. The initial velocity of the flow is denoted by $u_{0}(x)$.
The pressure is a scalar dynamic variable which is however determined
by the velocity (up to a function depending only on $t$). Indeed,
by taking the divergence of both sides of the first equation (\ref{3D-Ns01}),
i.e. applying $\frac{\partial}{\partial x_{i}}$ to this equation
and summing up $i=1,\ldots, d$, one obtains
\begin{equation}
\Delta P=-\sum_{j,i=1}^{d}\frac{\partial u^{j}}{\partial x_{i}}\frac{\partial u^{i}}{\partial x_{j}}+\nabla\cdot F\quad\textrm{ in }D,\label{3D-P-01}
\end{equation}
where we have used the divergence-free condition (\ref{3D-Ns02}).
The boundary value of $P$ remains to be determined. Since $u$ obeys
the no-slip condition, so by reading the first equation (\ref{3D-Ns01})
along the boundary $\partial D$ one obtains
\begin{equation}
\left.\nabla P\right|_{\partial D}=\nu\left.\Delta u\right|_{\partial D}+\left.F\right|_{\partial D}.\label{B-P-01}
\end{equation}

Instead of working out the boundary condition for $P$, we now consider the vorticity
$\omega=\nabla\wedge u$ whose components $\omega^{j}=\varepsilon^{jki}\frac{\partial}{\partial x_{k}}u^{i}$
when $d=3$ and $\omega=\frac{\partial}{\partial x_{1}}u^{2}-\frac{\partial}{\partial x_{2}}u^{1}$ when $d=2$,
which is in fact (up to a sign) the exterior derivative of $u$. Hence
by applying the linear differential operator $\varepsilon^{jki}\frac{\partial}{\partial x_{k}}$
to both sides of (\ref{3D-Ns01}), 
we shall obtain that
\begin{equation}
\frac{\partial}{\partial t}\omega+(u\cdot\nabla)\omega-\nu\Delta\omega-(\omega\cdot\nabla)u-G=0\quad\textrm{ in }D,\label{3D-vor-01}
\end{equation}
where $G=\nabla\wedge F$ with components $G^{j}=\varepsilon^{jki}\frac{\partial}{\partial x_{k}}F^{i}$ if $d=3$; if $d=2$, then $G=\frac{\partial}{\partial x_{1}}F^{2}-\frac{\partial}{\partial x_{2}}F^{1}$
and $(\omega\cdot\nabla)u=0$ identically. 

In order to utilize the vorticity transport equation (\ref{3D-vor-01}),
we need to identify the boundary values of $\omega$, i.e. the boundary
vorticity. Since $u$ obeys the non-slip condition, so that the normal
part of $\omega$ at the boundary $\omega^{\perp}=\nabla^{\Gamma}\wedge u^{\parallel}=0$,
where $u^{\Vert}$ denotes the tangential part of $u$ at the boundary, and $\nabla^{\Gamma}$ is the gradient operator in the boundary.
For identifying the tangential part of $\omega$, we notice that the
outwards unit normal $\boldsymbol{\nu}=-\frac{\partial}{\partial x_{3}}$.
Hence
\begin{equation}
\left.\omega^{1}\right|_{\partial D}=\left.\frac{\partial u^{3}}{\partial x_{2}}-\frac{\partial u^{2}}{\partial x_{3}}\right|_{\partial D}=-\left.\frac{\partial u^{2}}{\partial x_{3}}\right|_{\partial D}=-2\left.S_{23}\right|_{\partial D}\label{vort1-bv-01}
\end{equation}
and
\begin{equation}
\left.\omega^{2}\right|_{\partial D}=\left.\frac{\partial u^{1}}{\partial x_{3}}-\frac{\partial u^{3}}{\partial x_{1}}\right|_{\partial D}=\left.\frac{\partial u^{1}}{\partial x_{3}}\right|_{\partial D}=2\left.S_{13}\right|_{\partial D},\label{vort1-bv-02}
\end{equation}
where 
\[
S_{ij}=\frac{1}{2}\left(\frac{\partial u^{i}}{\partial x_{j}}+\frac{\partial u^{j}}{\partial x_{i}}\right)
\]
is the symmetric tensor field of rate-of-strain. Observe that the
normal part of the symmetric tensor field $S=(S_{ij})$, denoted by
$S^{\perp}$ is given by
\[
S^{\perp}=\left.(S_{13},S_{23},S_{33})\right|_{\partial D}.
\]
However, $\nabla\cdot u=0$, and $S_{11}=S_{22}=0$ on $\partial D$,
hence $S_{33}=0$ too. Therefore $S^{\perp}$ can be identified with
\[
S^{\perp}=\left.(S_{13},S_{23},0)\right|_{\partial D}
\]
at the boundary. Therefore the boundary vorticity $\left.\omega\right|_{\partial D}$,
denoted by $\theta$, is identified with twice of the stress at the
boundary
\begin{equation}
\theta=\left.2(-S_{23},S_{13},0)\right|_{\partial D}.\label{A-01}
\end{equation}
Therefore the vorticity $\omega$ is evolved according to the following
non-homogeneous boundary problem: 
\begin{equation}
\begin{cases}
\frac{\partial\omega}{\partial t}+(u\cdot\nabla)\omega-\nu\Delta\omega-(\omega\cdot\nabla)u-G=0 & \textrm{ in }D,\\
\left.\omega\right|_{\partial D}-\theta=0 & \textrm{ on }\partial D.
\end{cases}\label{tan-02}
\end{equation}

Note that the boundary vorticity $\theta$ is a tensor field on $\partial D$.
\begin{rem}
The boundary vorticity $\theta$ can not be determined a priori, which
causes a major problem for numerically computing solutions to the
boundary value problem of the Navier-Stokes equations via the random
vortex method (cf. \citep{AndersonGreengard1985, Chorin 1973, CottetKoumoutsakos2000, Goodman1987, Leonard1980, Majda and Bertozzi 2002, MarchioroPulvirenti1984}).
While some authors supply instead the vorticity equations \eqref{3D-vor-01}
with the Neumann boundary condition, which is in general not correct. 
\end{rem}

\section{Dynamics of the boundary vorticity}\label{Dynamics of the boundary vorticity}

In this section we shall derive the dynamical equation of the boundary
vorticity $\theta$ which is the trace $\left.\omega\right|_{\partial D}$
of the vorticity $\omega$ at the boundary. To this end we assume
that the velocity $u(x,t)$ is at least $C^{3}$ up to the boundary
$\partial D$. Since $u$ satisfies the non-slip condition, by reading
the vorticity equation (\ref{3D-vor-01}) along $\partial D$ we therefore obtain
\begin{equation}
\frac{\partial\theta}{\partial t}-\nu\left.\Delta\omega\right|_{\partial D}-\left.(\theta\cdot\nabla)u\right|_{\partial D}-\psi=0\quad\textrm{ in }\partial D,\label{BV-e-01}
\end{equation}
where $\psi = \left.G\right|_{\partial D}$, the boundary value of $G$. Using the non-slip condition
again, we deduce that $\theta^{3}=0$ and
\begin{equation}
\left.(\theta\cdot\nabla)u\right|_{\partial D}=\left.\theta^{1}\frac{\partial u}{\partial x_{1}}\right|_{\partial D}+\left.\theta^{2}\frac{\partial u}{\partial x_{2}}\right|_{\partial D}=0.\label{BV-e02}
\end{equation}
We therefore have a very important consequence.
\begin{thm}
At the boundary, two non-linear terms appearing in the vorticity transport
equation, the non-linear convection and the non-linear vorticity
stretching, neither of them participates directly in the generation of the vorticity
at the wall. 
\end{thm}

We are now in a position to state our main result of the paper.
\begin{thm}
\label{thm3}Let $D=\mathbb{R}_{+}^{3}$. Then $\theta^{3}=0$, and
$\theta^{1}$ and $\theta^{2}$ evolve according to the following
dynamics:
\begin{equation}
\begin{cases}
\frac{\partial\theta^{1}}{\partial t}-2\nu\Delta_{\Gamma}\theta^{1}+\nu\frac{\partial}{\partial x_{1}}\left(\nabla^{\Gamma}\cdot\theta\right)-\nu\left.\frac{\partial^{3}}{\partial\boldsymbol{\nu}^{3}}u^{2}\right|_{\partial D}-\psi^{1}=0\\
\frac{\partial\theta^{2}}{\partial t}-2\nu\Delta_{\Gamma}\theta^{2}+\nu\frac{\partial}{\partial x_{2}}\left(\nabla^{\Gamma}\cdot\theta\right)+\nu\left.\frac{\partial^{3}}{\partial\boldsymbol{\nu}^{3}}u^{1}\right|_{\partial D}-\psi^{2}=0
\end{cases}\label{th-00}
\end{equation}
in $\partial D=\mathbb{R}^{2}$. That is
\begin{equation}
\frac{\partial\theta}{\partial t}-2\nu\Delta_{\Gamma}\theta+\nu\nabla^{\Gamma}\left(\nabla^{\Gamma}\cdot\theta\right)+\nu\star\left.\frac{\partial^{3}}{\partial\boldsymbol{\nu}^{3}}u^{\parallel}\right|_{\partial D}-\psi=0.\label{g-bvort-03}
\end{equation}
Here $\boldsymbol{\nu}$ is the normal to $\partial D$ pointing outward, i.e. $\boldsymbol{\nu}=-\frac{\partial}{\partial x_{3}}$,
$\Delta_{\Gamma}$ and $\nabla^{\Gamma}$ denote the Laplacian and
gradient operator on $\mathbb{R}^{2}$ respectively. Here $\star$ is the
Hodge star operator of $\partial D$, and $u^{\parallel}$ is the tangential extension, in this case,  $u^{\parallel}=(u^1,u^2)$.
\end{thm}

Before we give the derivation of the boundary vorticity dynamics,
we would like to make several comments.
\begin{rem}
The dynamical equations \eqref{th-00} imply that the kinematic viscosity
constant at the boundary is exactly doubled, as if the fluid became more `viscous'
than the fluid in the main stream. This phenomenon is actually true
for any viscous wall-bounded flow constrained by a curved solid wall.
\end{rem}

\begin{rem}
The motion equation \eqref{g-bvort-03} for the boundary vorticity
also indicates clearly how the external flow (i.e., the flow away from
the boundary) participates in the generation of the vorticity at the
boundary. More precisely, the boundary vorticity is generated, 
with the help of the initial boundary vorticity and the external boundary force
$\psi$, together with an `external' force $-\nu\star\left.\frac{\partial^{3}}{\partial\boldsymbol{\nu}^{3}}u^{\parallel}\right|_{\partial D}$
from the main stream flow exerted on the "self-dynamics" of the boundary vorticity, which is determined
by the linear heat operator
\[
\frac{\partial\theta}{\partial t}-2\nu\Delta_{\Gamma}\theta+\nu\nabla^{\Gamma}\left(\nabla^{\Gamma}\cdot\theta\right).
\]
\end{rem}

\begin{rem}
For a typical wall-bounded viscous fluid flow, in particular for turbulent
boundary layer flows, the boundary vorticity $(\theta^{1},\theta^{2})$
(which equals the normal stress at the boundary) is significant, which
is big in comparison with the typical scale of the flow. While the
`external' force inherited from the outer layer flow, which adjusts
the self-dynamics of the boundary vorticity, is proportional to the
kinematic viscosity $\nu$. Since the dynamical equation 
\[
\frac{\partial\tilde{\theta}}{\partial t}-2\nu\Delta_{\Gamma}\tilde{\theta}+\nu\nabla^{\Gamma}\left(\nabla^{\Gamma}\cdot\tilde{\theta}\right)-\psi=0
\]
subject to the same initial boundary vorticity $\tilde{\theta}=\theta$
at $t=0$, is linear, hence if $\left.\frac{\partial^{3}}{\partial\boldsymbol{\nu}^{3}}u^{\parallel}\right|_{\partial D}$
is bounded and the kinematic viscosity $\nu$ is small, then the boundary
vorticity $\theta(x,t)$ is more or less self-organised, and the outer layer
flow inserts insignificant impact on the generation of the boundary
vorticity.
\end{rem}

\begin{proof}
{[}of Theorem \ref{thm3}{]} The proof is completely elementary.
We begin with Eq. (\ref{BV-e-01}) and we need to compute the trace
of $\Delta\omega$ at the boundary. While it is clear that
\[
\left.\Delta\omega^{i}\right|_{\partial D}=\left.\left(\frac{\partial^{2}}{\partial x_{1}^{2}}+\frac{\partial^{2}}{\partial x_{2}^{2}}\right)\omega^{i}+\frac{\partial^{2}}{\partial x_{3}^{2}}\omega^{i}\right|_{\partial D}=\Delta_{\Gamma}\theta^{i}+\left.\frac{\partial^{2}}{\partial x_{3}^{2}}\omega^{i}\right|_{\partial D}
\]
where the last term has to be computed. While
\begin{align*}
\frac{\partial}{\partial x_{3}}\omega^{1} & =\frac{\partial}{\partial x_{3}}\left(\frac{\partial u^{3}}{\partial x_{2}}-\frac{\partial u^{2}}{\partial x_{3}}\right)=\frac{\partial}{\partial x_{3}}\frac{\partial u^{3}}{\partial x_{2}}-\frac{\partial}{\partial x_{3}}\frac{\partial}{\partial x_{3}}u^{2}\\
 & =-\frac{\partial}{\partial x_{2}}\left(\frac{\partial u^{1}}{\partial x_{1}}+\frac{\partial u^{2}}{\partial x_{2}}\right)-\frac{\partial}{\partial x_{3}}\frac{\partial}{\partial x_{3}}u^{2}
\end{align*}
and therefore
\begin{align*}
\frac{\partial^{2}}{\partial x_{3}^{2}}\omega^{1} & =\frac{\partial^{2}}{\partial x_{3}^{2}}\left(\frac{\partial}{\partial x_{2}}u^{3}-\frac{\partial}{\partial x_{3}}u^{2}\right)\\
 & =\frac{\partial^{2}}{\partial x_{2}\partial x_{3}}\frac{\partial}{\partial x_{3}}u^{3}-\frac{\partial^{2}}{\partial x_{3}^{2}}\frac{\partial}{\partial x_{3}}u^{2}\\
 & =-\frac{\partial^{2}}{\partial x_{2}\partial x_{3}}\frac{\partial}{\partial x_{1}}u^{1}-\frac{\partial^{2}}{\partial x_{2}\partial x_{3}}\frac{\partial}{\partial x_{2}}u^{2}-\frac{\partial^{2}}{\partial x_{3}^{2}}\frac{\partial}{\partial x_{3}}u^{2}\\
 & =\frac{\partial^{2}}{\partial x_{2}^{2}}\omega^{1}+\frac{\partial^{2}}{\partial x_{1}^{2}}\omega^{1}-\frac{\partial^{2}}{\partial x_{1}^{2}}\omega^{1}-\frac{\partial^{2}}{\partial x_{2}\partial x_{1}}\omega^{2}-\frac{\partial^{2}}{\partial x_{3}^{2}}\frac{\partial}{\partial x_{3}}u^{2}\\
 & =\Delta_{\Gamma}\theta^{1}-\frac{\partial}{\partial x_{1}}\left(\nabla_{\Gamma}\cdot\theta\right)-\frac{\partial^{2}}{\partial x_{3}^{2}}\frac{\partial}{\partial x_{3}}u^{2}
\end{align*}
It follows that
\[
\left.\frac{\partial^{2}}{\partial x_{3}^{2}}\omega^{1}\right|_{\partial D}=-\frac{\partial^{3}}{\partial x_{3}^{3}}u^{2}+\Delta_{\Gamma}\theta^{1}-\frac{\partial}{\partial x_{1}}\left(\nabla_{\Gamma}\cdot\theta\right).
\]
Similarly
\begin{align*}
\frac{\partial^{2}}{\partial x_{3}^{2}}\omega^{2} & =\frac{\partial^{2}}{\partial x_{3}^{2}}\left(\frac{\partial}{\partial x_{3}}u^{1}-\frac{\partial}{\partial x_{1}}u^{3}\right)\\
 & =-\frac{\partial^{2}}{\partial x_{1}\partial x_{3}}\frac{\partial}{\partial x_{3}}u^{3}+\frac{\partial^{2}}{\partial x_{3}^{2}}\frac{\partial}{\partial x_{3}}u^{1}\\
 & =\frac{\partial^{2}}{\partial x_{1}\partial x_{3}}\frac{\partial}{\partial x_{1}}u^{1}+\frac{\partial^{2}}{\partial x_{1}\partial x_{3}}\frac{\partial}{\partial x_{2}}u^{2}+\frac{\partial^{2}}{\partial x_{3}^{2}}\frac{\partial}{\partial x_{3}}u^{1}\\
 & =\frac{\partial^{2}}{\partial x_{1}^{2}}\omega^{2}+\frac{\partial^{2}}{\partial x_{2}^{2}}\omega^{2}-\frac{\partial^{2}}{\partial x_{2}^{2}}\omega^{2}-\frac{\partial^{2}}{\partial x_{2}\partial x_{1}}\omega^{1}+\frac{\partial^{2}}{\partial x_{3}^{2}}\frac{\partial}{\partial x_{3}}u^{1}\\
 & =\Delta_{\Gamma}\theta^{2}-\frac{\partial}{\partial x_{2}}\left(\nabla_{\Gamma}\cdot\theta\right)+\frac{\partial^{2}}{\partial x_{3}^{2}}\frac{\partial}{\partial x_{3}}u^{1}
\end{align*}
where the second equality follows from the divergence-free: $\nabla\cdot u=0$.
Hence
\begin{equation}
\left.\frac{\partial}{\partial x_{3}}\omega^{1}\right|_{\partial D}=-\left.\frac{\partial}{\partial x_{3}}\frac{\partial}{\partial x_{3}}u^{2}\right|_{\partial D}\label{ome-01}
\end{equation}
and
\begin{equation}
\left.\frac{\partial}{\partial x_{3}}\omega^{2}\right|_{\partial D}=\left.\frac{\partial}{\partial x_{3}}\frac{\partial}{\partial x_{3}}u^{1}\right|_{\partial D}.\label{ome-02}
\end{equation}
so that
\[
\left.\frac{\partial^{2}}{\partial x_{3}^{2}}\omega^{2}\right|_{\partial D}=\frac{\partial^{3}}{\partial x_{3}^{3}}u^{1}+\Delta_{\Gamma}\theta^{1}-\frac{\partial}{\partial x_{2}}\left(\nabla_{\Gamma}\cdot\theta\right).
\]
Similarly 
\begin{align*}
\frac{\partial}{\partial x_{3}}\omega^{3} & =\frac{\partial}{\partial x_{3}}\left(\frac{\partial u^{2}}{\partial x_{1}}-\frac{\partial u^{1}}{\partial x_{2}}\right)=\frac{\partial}{\partial x_{3}}\frac{\partial u^{2}}{\partial x_{1}}-\frac{\partial}{\partial x_{3}}\frac{\partial}{\partial x_{2}}u^{1}\\
 & =\frac{\partial}{\partial x_{1}}\frac{\partial u^{2}}{\partial x_{3}}-\frac{\partial}{\partial x_{2}}\frac{\partial u^{1}}{\partial x_{3}}\\
 & =-\left(\frac{\partial}{\partial x_{1}}\omega^{1}+\frac{\partial}{\partial x_{2}}\omega^{2}\right)-\frac{\partial}{\partial x_{2}}\frac{\partial u^{3}}{\partial x_{1}}+\frac{\partial}{\partial x_{1}}\frac{\partial u^{3}}{\partial x_{2}}
\end{align*}
and
\[
\frac{\partial^{2}}{\partial x_{3}^{2}}\omega^{3}=-\frac{\partial}{\partial x_{3}}\left(\frac{\partial}{\partial x_{1}}\omega^{1}+\frac{\partial}{\partial x_{2}}\omega^{2}\right)
\]
so that
\begin{equation}
\left.\frac{\partial}{\partial x_{3}}\omega^{3}\right|_{\partial D}=-\nabla^{\Gamma}\cdot\theta=-\frac{\partial\theta^{1}}{\partial x_{1}}-\frac{\partial\theta^{2}}{\partial x_{2}}.\label{ome-03}
\end{equation}
Putting these equations together we obtain (\ref{th-00}).
\end{proof}
For convenience let us write down the evolution for two dimensional
flows for reference below. 
\begin{thm}\label{boundary vorticity thm 2d}
If $d=2$ (so that both $\omega=\frac{\partial}{\partial x_{1}}u^{2}-\frac{\partial}{\partial x_{2}}u^{1}$
and its trace $\theta$ at the boundary are scalar function), then
the boundary vorticity $\theta$ evolves according to the following
dynamical equation
\begin{equation}
\frac{\partial\theta}{\partial t}-2\nu\Delta_{\Gamma}\theta-\nu\left(\frac{\partial}{\partial\boldsymbol{\nu}}\right)^{3}u^{\Vert}-\psi=0\label{2D-evo-01}
\end{equation}
where $u^{\Vert}=u^{1}$ is the tangent component of the velocity
field $u$.
\end{thm}

\begin{proof}
In this case we prefer use coordinates $x=(x_{1},x_{2})$. For 2D
flow, the vorticity transport equation becomes
\begin{equation}
\frac{\partial\omega}{\partial t}+(u\cdot\nabla)\omega-\nu\Delta\omega=G\quad\textrm{ in }D\label{2D-ome-01}
\end{equation}
where $\omega=\frac{\partial u^{2}}{\partial x_{1}}-\frac{\partial u^{1}}{\partial x_{2}}$,
so that
\begin{align*}
\Delta\omega & =\frac{\partial^{2}}{\partial x_{1}^{2}}\omega+\frac{\partial^{2}}{\partial x_{2}^{2}}\left(\frac{\partial u^{2}}{\partial x_{1}}-\frac{\partial u^{1}}{\partial x_{2}}\right)\\
 & =2\frac{\partial^{2}}{\partial x_{1}^{2}}\omega-\frac{\partial^{3}}{\partial x_{1}^{3}}u^{2}-\frac{\partial^{3}}{\partial x_{2}^{3}}u^{1}
\end{align*}
so that 
\[
\left.\Delta\omega\right|_{\partial D}=2\Delta_{\Gamma}\theta+\left(\frac{\partial}{\partial\boldsymbol{\nu}}\right)^{3}u^{1}
\]
and the conclusion follows immediately.
\end{proof}

\section{Functional integral representations}\label{Functional integral representations}

In the next two sections we demonstrate the use of the dynamical
equations in the stochastic direct numerical simulations of
the viscous flows within thin layers next to the fluid boundary.

We shall develop random vortex schemes for calculating
numerically solutions to the boundary problem (\ref{3D-Ns01}, \ref{3D-Ns02}) by using the dynamical equations of the boundary vorticity for
updating the boundary values of the vorticity in numerical schemes.

To exhibit our ideas clearly we deal with 2D flows only, i.e. $d=2$
and $D=\{x:x_{2}>0\}$. Since the boundary vorticity $\theta$ is
in general non-trivial, so we introduce a family of perturbations
of $\omega$ defined by $W^{\varepsilon}=\omega-\sigma_{\varepsilon}$
for every $\varepsilon>0$, given by
\begin{equation}
\sigma_{\varepsilon}(x_{1},x_{2},t)=\theta(x_{1},t)\phi(x_{2}/\varepsilon),\label{ext-01}
\end{equation}
where $\phi:[0,\infty)\rightarrow[0,1]$ is a proper cut-off function
such that $\phi(r)=1$ for $r\in[0,1/3)$ and $\phi(r)=0$ for $r\geq2/3$.
Indeed we will use the following cut-off function: 
\begin{equation}
\phi(r)=\begin{cases}
1 & \textrm{for \ensuremath{r\in[0,1/3)},}\\
\frac{1}{2}+54\left(r-\frac{1}{2}\right)^{3}-\frac{9}{2}\left(r-\frac{1}{2}\right) & \textrm{ for }r\in[1/3,2/3],\\
0 & \textrm{ for }r\geq2/3
\end{cases}\label{phi-def}
\end{equation}
Hence $-54\leq\phi''\leq54$, $-\frac{9}{2}\leq\phi'\leq0$ on $[1/3,2/3]$
and $\phi'=0$ for $r\leq1/3$ or $r\geq2/3$. In fact 
\begin{equation}
\phi'(r)=\begin{cases}
162\left(r-\frac{1}{2}\right)^{2}-\frac{9}{2} & \textrm{ for }r\in[1/3,2/3],\\
0 & \textrm{ otherwise }
\end{cases}\label{phi-1d}
\end{equation}
and 
\begin{equation}
\phi''(r)=\begin{cases}
324\left(r-\frac{1}{2}\right) & \textrm{ for }r\in[1/3,2/3],\\
0 & \textrm{ otherwise. }
\end{cases}\label{phi-2d}
\end{equation}

Then $W_{\varepsilon}$ is the solution to the following Dirichlet
boundary problem of the parabolic equation: 
\begin{equation}
\left(\frac{\partial}{\partial t}+u\cdot\nabla-\nu\Delta\right)W_{\varepsilon}-g_{\varepsilon}=0\quad\textrm{ in }D,\quad\textrm{ and }\left.W_{\varepsilon}\right|_{\partial D}=0,\label{eq:qq4}
\end{equation}
where 
\begin{align}
g_{\varepsilon}(x,t) & =G(x,t)+\frac{\nu}{\varepsilon^{2}}\phi''(x_{2}/\varepsilon)\theta(x_{1},t)-\frac{1}{\varepsilon}\phi'(x_{2}/\varepsilon)u^{2}(x,t)\theta(x_{1},t)\nonumber \\
 & +\phi(x_{2}/\varepsilon)\left(\nu\frac{\partial^{2}\theta}{\partial x_{1}^{2}}(x_{1},t)-\frac{\partial\theta}{\partial t}(x_{1},t)\right)-\phi(x_{2}/\varepsilon)u^{1}(x,t)\frac{\partial\theta}{\partial x_{1}}(x_{1},t)\label{g-xt-01}
\end{align}
for any $x=(x_{1},x_{2})$, $x_{2}>0$. The initial data for $W^{\varepsilon}$
is given by 
\begin{equation}
W_{0}^{\varepsilon}(x)=\omega_{0}(x_{1},x_{2})-\omega_{0}(x_{1},0)\phi(x_{2}/\varepsilon)\quad\textrm{ for }x\in D.\label{int-c1}
\end{equation}

We shall need the stochastic integral representation
in terms of the Taylor diffusion with velocity $u(x,t)$. To this
end, the vector field $u(x,t)$ is extended to a vector field on $\mathbb{R}^{2}$
by reflection about the line $x_{2}=0$ so that 
\[
u^{1}(x,t)=u^{1}(\overline{x},t),\quad u^{2}(x,t)=-u^{2}(\overline{x},t)
\]
for $x=(x_{1},x_{2})$ with $x_{2}>0$, here $x\mapsto\overline{x}$
is the reflection about the line that $x_{2}=0$, that is, $\overline{x}=(x_{1},-x_{2})$
for $x=(x_{1},x_{2})\in\mathbb{R}^{2}$. This extension retains the
divergence-free property, though, in distribution. That is, $\nabla\cdot u(\cdot,t)=0$
on $\mathbb{R}^{2}$ in the sense of distribution.

For each $\xi\in\mathbb{R}^{2}$, $(X_{t}^{\xi})_{t\geq0}$ is the
unique (weak) solution of the following It\^o's stochastic differential
equation 
\begin{equation}
\textrm{d}X_{t}^{\xi}=u(X_{t}^{\xi},t)\textrm{d}t+\sqrt{2\nu}\textrm{d}B_{t},\quad X_{0}^{\xi}=\xi,\label{X-sde1}
\end{equation}
where $B=(B_{t})$ is a two dimensional Brownian motion on some probability
space. Let $p(s,\xi;t,y)$ be the transition probability density function
of the diffusion $(X_{t}^{\xi})_{t\geq0}$, i.e. 
\[
p(s,x;t,y)\textrm{d}y=\mathbb{P}\left[\left.X_{t}^{\xi}\in\textrm{d}y\right|X_{s}^{\xi}=x\right]
\]
for $t>s\geq0$ and $x,y\in\mathbb{R}^{2}$ (which is independent
of $\xi$). Let $p^{D}(s,x;t,y)$ be the transition (sub-)probability
density function of the diffusion $X^{\xi}$ killed on leaving the
region $D$, where $t>s\geq0$, $x,y\in D$. Then 
\[
p^{D}(s,x;t,y)\textrm{d}y=\mathbb{P}\left[\left.1_{\left\{ \zeta(X^{\eta})>t\right\} },X_{t}^{\xi}\in\textrm{d}y\right|X_{s}^{\xi}=x\right]
\]
for any $t>s\geq0$, and
\begin{equation}
p^{D}(s,x;t,y)=p(s,x;t,y)-p(s,x;t,\bar{y})\label{f-04}
\end{equation}
for $t>s\geq0$ and $x,y\in D$, where $\zeta(\psi)=\inf\left\{ t:\psi(t)\notin D\right\} $.
Note that, since $\overline{u(x,t)}=u(\bar{x},t)$ for $x\in\mathbb{R}^{2}$
and $t\geq0$, $p(s,x;t,y)=p(s,\bar{x};t,\bar{y})$.
\begin{thm}
For every $\varepsilon>0$, it holds that 
\begin{align}
W_{\varepsilon}(y,t)= & \int_{D}\mathbb{P}\left[\left.\zeta(X^{\eta})>t\right|X_{t}^{\eta}=y\right]W_{\varepsilon}(\eta,0)p(0,\eta;t,y)\textrm{d}\eta\nonumber \\
 & +\int_{0}^{t}\int_{D}\mathbb{E}\left[\left.1_{\left\{ s>\gamma_{t}(X^{\eta})\right\} }g_{\varepsilon}(X_{s}^{\eta},s)\right|X_{t}^{\eta}=y\right]p(0,\eta;t,y)\textrm{d}\eta\textrm{d}s\label{rep-m02-1}
\end{align}
for every $t>0$ and $y\in D$, where $\gamma_{t}(\psi)=\sup\left\{ s\in(0,t):\psi(s)\notin D\right\} $
for every continuous path $\psi$. 
\end{thm}

For a proof of this representation, we refer to \citep{Qian2022}
and \citep{LiQianXu2023}. We emphasize that the previous representation
(\ref{rep-m02-1}) is different from the solution representation
in terms of the fundamental solution in that only the Taylor diffusion
starting at a fixed time $0$ is required, which therefore reduces
the computational cost substantially when numerical schemes are implemented
based on such integral representations. 

We next establish a representation for $u(x,t)$ by applying the Biot-Savart
law. To this end, we apply the following convention. For 2D vectors,
the following convention, which is consistent with the canonical identifications
with 3D vectors, will be adopted. If $a=(a_{1},a_{2})$ and $b=(b_{1},b_{2})$,
then $a\wedge b=a_{1}b_{2}-a_{2}b_{1}$ (a scalar), and if $c$ is
a scalar, then $a\wedge c=(a_{2}c,-a_{1}c)$.
\begin{thm}
The following stochastic integral representation for the velocity
holds: 
\begin{align}
u(x,t) & =\int_{D}K(y,x)\wedge\sigma_{\varepsilon}(y,t)\textrm{d}y\nonumber \\
 & +\int_{D}\mathbb{E}\left[1_{D}(X_{t}^{\xi})K(X_{t}^{\xi},x)-1_{D}(X_{t}^{\bar{\xi}})K(X_{t}^{\bar{\xi}},x)\right]\wedge W_{\varepsilon}(\xi,0)\textrm{d}\xi\nonumber \\
 & +\int_{0}^{t}\int_{D}\mathbb{E}\left[\left.1_{\left\{ s>\gamma_{t}(X^{\eta})\right\} }K(X_{t}^{\eta},x)\wedge g_{\varepsilon}(X_{s}^{\eta},s)\right|\right]\textrm{d}\eta\textrm{d}s\label{main formula}
\end{align}
for every $x\in D$, and $u(x,t)=\overline{u(\overline{x},t)}$ for
$x_{2}<0$, and $u(x,t)=0$ if $x_{2}=0$.
\end{thm}

\begin{proof}
Recall that the Biot-Savart singular integral kernel for $D$ (which
is the gradient of the green function for $D$) is given by 
\begin{equation}
K(y,x)=\frac{1}{2\pi}\left(\frac{y-x}{|y-x|^{2}}-\frac{y-\overline{x}}{|y-\overline{x}|^{2}}\right)\label{eq:qq13-1}
\end{equation}
for $y\neq x$ or $\overline{x}$. Since $\nabla\cdot u=0$, $\nabla\wedge u=\omega$
and $u$ is subject to the Dirichlet boundary condition that $u(x,t)=0$
for $x\in\partial D$, hence, according to Green formula we obtain
that 
\begin{equation}\label{BSLaw}
u(x,t)=\int_{D}K(y,x)\wedge\omega(y,t)\textrm{d}y.
\end{equation}
While by definition, for every $\varepsilon>0$, $\omega=\sigma_{\varepsilon}+W_{\varepsilon}$,
the representation follows by utilising the representation (\ref{rep-m02-1})
and the Fubini theorem. 
\end{proof}

\section{Numerical experiments}\label{Numerical experiments}

In this section, we provide some numerical simulations for the representations discussed above, focusing on the two-dimensional case. Recall that, by letting $\epsilon\downarrow 0$ in \eqref{main formula} the following representation holds:
\begin{align}
u(x,t) & =\int_{D} K^{\perp}(x,\eta) \sigma_{\varepsilon}(\eta,t)\textrm{d}\eta+\int_{D}\mathbb{E}\left[K^{\perp}(x,X_{t}^{\eta})1_{\left\{ t<\zeta(X^{\eta}\circ\tau_{t})\right\} }\right]\omega_0(\eta)\textrm{d}\eta\nonumber \\
 & -\int_{D}\mathbb{E}\left[K^{\perp}(x,X_{t}^{\eta})1_{\left\{ t<\zeta(X^{\eta}\circ\tau_{t})\right\} }\right]\sigma_{\varepsilon}(\eta,t)\textrm{d}\eta\nonumber \\
 & +\int_{0}^{t}\int_{D}\mathbb{E}\left[1_{\{t-s<\zeta(X^{\eta}\circ\tau_{t})\}}K^{\perp}(x,X_{t}^{\eta})G(X_{s}^{\eta},s)\right]\textrm{d}\eta\textrm{d}s\nonumber\\
 & +\int_{0}^{t}\int_{D}\mathbb{E}\left[1_{\{t-s<\zeta(X^{\eta}\circ\tau_{t})\}}K^{\perp}(x,X_{t}^{\eta})\rho_{\varepsilon}(X_{s}^{\eta},s)\right]\textrm{d}\eta\textrm{d}s,
\end{align}
where $\rho_{\varepsilon} = g_{\varepsilon} - G$ and \begin{equation}\label{KernelPerp}
K^{\perp}=(K^2,-K^1)
\end{equation} in our notation. In the following we ease the notation by omitting the superscript in the kernel and the wedge product coming from the Biot-Savart law \eqref{BSLaw}, which is essentially equivalent to redefining the kernel as in \eqref{KernelPerp}. 

The only term in the definition of $\rho_{\varepsilon}(x,t)$, dependent on $\varepsilon$, that does not vanish in the limit is 
\begin{equation}\label{boundary term}
\frac{\nu}{\varepsilon^{2}}\phi''(x_{2}/\varepsilon)\theta(x_{1},t).
\end{equation}
Therefore, one can approximate the representation for the velocity $u$ by the following
\begin{align}
u^{i}(x,t) & \approx\int_{D}\mathbb{E}\left[1_{\left\{ t<\zeta(X^{\eta}\circ\tau_{t})\right\} }K^{i}(x,X_{t}^{\xi})\right]\omega_0(\xi)\textrm{d}\xi\nonumber\\& +\int_{0}^{t}\int_{D}\mathbb{E}\left[1_{\{t-s<\zeta(X^{\xi}\circ\tau_{t})\}}K^{i}(x,X_{t}^{\xi})G(X_{s}^{\xi},s)\right]\textrm{d}\xi\textrm{d}s\nonumber\\
 &+\frac{\nu}{\varepsilon^2}\int_{0}^{t}\int_{D}\mathbb{E}\left[1_{\{t-s<\zeta(X^{\xi}\circ\tau_{t})\}}K^{i}(x,X_{t}^{\xi})\theta(X_{s}^{\xi}, s) \phi'' (X_{s}^{\xi}/\varepsilon)\right]\textrm{d}\xi\textrm{d}s,  
\end{align}
for some small $\varepsilon$. That is, we omit the terms that do not contribute to the limit and \eqref{boundary term} is approximated by taking sufficiently small $\varepsilon$. Notice that as the support of $\phi''$ is the interval $[1/3,2/3]$, the last integration can be taken over a thin layer close to the boundary. Notice also that in the last integration we used $\theta(x,s), \phi''(x/\varepsilon)$ to denote $\theta(x_1,s), \phi''(x_2/\varepsilon)$ slightly abusing notation.

We use the idea described above to the representation \eqref{main formula}, i.e. we approximate the velocity similarly as
\begin{align}\label{velocity approx}
u(x,t) & \approx\int_{D}\mathbb{E}\left[1_{D}(X_{t}^{\xi})K(X_{t}^{\xi},x)-1_{D}(X_{t}^{\bar{\xi}})K(X_{t}^{\bar{\xi}},x)\right] \omega_0(\xi)\textrm{d}\xi\nonumber \\
 & +\int_{0}^{t}\int_{D}\mathbb{E}\left[1_{\left\{ s>\gamma_{t}(X^{\eta})\right\} }K(X_{t}^{\eta},x) G(X_{s}^{\eta},s)\right]\textrm{d}\eta\textrm{d}s\nonumber \\
  & +\frac{\nu}{\varepsilon^2}\int_{0}^{t}\int_{D}\mathbb{E}\left[1_{\left\{ s>\gamma_{t}(X^{\eta})\right\} }K(X_{t}^{\eta},x) \theta(X_{s}^{\eta}, s) \phi''(X_{s}^{\eta}/\varepsilon) \right]\textrm{d}\eta\textrm{d}s,
\end{align}
for every $x\in D$, and $u(x,t)=\overline{u(\overline{x},t)}$ for
$x_{2}<0$, and $u(x,t)=0$ if $x_{2}=0$. 

For the half-plane domain $D$, we introduce lattice points as follows. Notice that as in \eqref{velocity approx} the first integral contains processes with reflected initial positions $\bar{\xi}$, we have to add reflected lattice points for the below discretisation. 
\begin{enumerate}
    \item The thin boundary layer lattice $D_b$ is given by 
    \begin{equation}\label{PlateBLayerLattice}
        x_{b}^{i_1 i_2} = (i_1 h_1, i_2 h_2),\quad \text{ for } -N_1 \leq i_1 \leq N_1 \text{ and } -N_2 \leq i_2 \leq N_2,
    \end{equation}
    where $h_1, h_2$ are mesh sizes and $N_1, N_2$ are numbers of points. 

    \item The outer layer lattice $D_o$ is defined as
    \begin{equation}\label{PlateOLayerLattice}
        x_{o}^{i_1 i_2} = (i_1 h_0 , i_2 h_0),\quad \text{ for } -N_0 \leq i_1 \leq N_0 \text{ and } -N_2 \leq i_2 \leq N_2,
    \end{equation}
    where $h_0$ is mesh size and $N_0$ is the number of points. 
\end{enumerate}

The discretised random vortex system is described as follows. We initialise the processes $X^{i_1, i_2}_{b; t_0} = x_{b}^{i_1 i_2}$ and $X^{i_1, i_2}_{o; t_0} = x_{o}^{i_1 i_2}$ and update them for $k \geq 0$ according to
\begin{equation}\label{UpdX}
    X^{i_1, i_2}_{t_{k+1}} = X^{i_1, i_2}_{t_k} + h \Hat{u}(X^{i_1, i_2}_{t_k}, t_k) + \sqrt{2\nu} (B_{t_{k+1}}-B_{t_k}),
\end{equation}
where $t_k = kh$ for $k \geq 0$ and some fixed time mesh $h$. To ease the notation, we drop the subscripts $o$ and $b$. The processes are coupled with the drift $\Hat{u}$ which is given by
\begin{align}\label{HatURepr}
    \Hat{u}(x,t_{k+1}) &= \sum_{\substack{(i_1,i_2)\in D \\ i_2 > 0}} A_{i_1,i_2} \omega_{i_1,i_2} \mathbb{E}\left[1_{D}(X^{i_1, i_2}_{t_k})K(X^{i_1, i_2}_{t_k},x)-1_{D}(X^{i_1, -i_2}_{t_k})K(X^{i_1, -i_2}_{t_k},x)\right] \nonumber \\
    &+ \sum_{\substack{(i_1,i_2)\in D \\ i_2 > 0}} \sum_{l=0}^k A_{i_1,i_2} h G_{i_1,i_2; t_l} \mathbb{E} \left[1_{\left\{ t_l>\gamma_{t_k}(X^{i_1, i_2})\right\}} K(X^{i_1, i_2}_{t_k},x) G(X^{i_1, i_2}_{t_l},t_l) \right] \nonumber \\ 
    &+ \frac{\nu}{\varepsilon^2} \sum_{\substack{(i_1,i_2)\in D_b \\ i_2 > 0}} \sum_{l=0}^k h_1 h_2 h \mathbb{E} \left[ 1_{\left\{ t_l>\gamma_{t_k}(X^{i_1, i_2})\right\}} K(X^{i_1, i_2}_{t_k},x)\theta(X^{i_1, i_2}_{t_l}, t_l) \phi''(X^{i_1, i_2}_{t_l}/\varepsilon) \right], 
\end{align}
for $x \in D$, and $\Hat{u}(x,t)=\overline{\Hat{u}(\overline{x},t)}$ for
$x_{2}<0$, and $\Hat{u}(x,t)=0$ if $x_{2}=0$. In what follows, we unify summations over $(i_1,i_2)\in D_{o}$ and $(i_1,i_2)\in D_{b}$ writing summation over $(i_1,i_2)\in D$ with 
\begin{align}\label{a-d032}
    A_{i_1,i_2} &= h_1 h_2 \text{ or } h_0^2, \nonumber \\
    \omega_{i_1,i_2} &= \omega(x_{b}^{i_1 i_2}, 0) \text{ or } \omega(x_{o}^{i_1 i_2}, 0), \nonumber \\
    G_{i_1,i_2; t_l} &= G(x_{b}^{i_1 i_2}, t_l) \text{ or } G(x_{o}^{i_1 i_2}, t_l),
\end{align}
for boundary and outer layers. 

We conduct experiments using the following numerical scheme. To deal with expectations in the representation \eqref{HatURepr}, we drop them and run Brownian motions independent at each site $(i_1, i_2)$ in \eqref{UpdX}. Therefore, we update the diffusions $X_{t_k}^{i_1, i_2}$, starting at $x^{i_1 i_2}$ when $k = 0$, according to
\begin{equation}\label{NS1X}
    X^{i_1, i_2}_{t_{k+1}} = X^{i_1, i_2}_{t_k} + h \Hat{u}(X^{i_1, i_2}_{t_k}, t_k) + \sqrt{2\nu} (B_{t_{k+1}}^{i_1, i_2}-B_{t_k}^{i_1, i_2}),
\end{equation}
for $k \geq 0$, where $B^{i_1, i_2}$ are independent Brownian motions. The drift is given as 
    \begin{align}\label{NS1U}
    \Hat{u}(x,t_{k+1}) &= \sum_{\substack{(i_1,i_2)\in D \\ i_2 > 0}} A_{i_1,i_2} \omega_{i_1,i_2} \left(1_{D}(X^{i_1, i_2}_{t_k})K(X^{i_1, i_2}_{t_k},x)-1_{D}(X^{i_1, -i_2}_{t_k})K(X^{i_1, -i_2}_{t_k},x)\right) \nonumber \\
    &+ \sum_{\substack{(i_1,i_2)\in D \\ i_2 > 0}} A_{i_1,i_2} h K(X^{i_1, i_2}_{t_k},x) \sum_{l=0}^k 1_{\left\{ t_l>\gamma_{t_k}(X^{i_1, i_2})\right\}}G_{i_1,i_2; t_l} \nonumber \\ 
    &+ \frac{\nu}{\varepsilon^2} \sum_{\substack{(i_1,i_2)\in D_b \\ i_2 > 0}}  h_1 h_2 h K(X^{i_1, i_2}_{t_k},x) \sum_{l=0}^k 1_{\left\{ t_l>\gamma_{t_k}(X^{i_1, i_2})\right\}} \theta(X^{i_1, i_2}_{t_l}, t_l) \phi''(X^{i_1, i_2}_{t_l}/\varepsilon),
    \end{align}
    for $x \in D$ and $\Hat{u}(x,t)=\overline{\Hat{u}(\overline{x},t)}$ for
$x_{2}<0$, and $\Hat{u}(x,t)=0$ if $x_{2}=0$, with $A_{i_1,i_2}$, $ \omega_{i_1,i_2}$, $G_{i_1, i_2; t_l}$ given in \eqref{a-d032}. Notice that also in practice instead of the kernel $K$, we compute a regularised version denoted by $K_{\delta}$, e.g., $K_{\delta}(y,x) = K(y,x) \left( 1 - \exp\left(-|y-x|^2 / \delta \right) \right)$.

The above representation \eqref{NS1U} for the velocity depends on the boundary vorticity $\theta$. In \citep{CherepanovQian2023}, the derivative of \eqref{NS1U} with respect to $x$ was used to compute $\theta$ which is possible if the kernel $K$ is replaced with a mollified integral kernel of  $K_{\delta}$. Here we shall use a different approach --- recall from Theorem \ref{boundary vorticity thm 2d} that the boundary vorticity $\theta(x_1,t)$ solves the equation
\begin{equation*}
\frac{\partial}{\partial t}\theta-2 \nu \frac{\partial^2}{\partial x_1^2} \theta = \psi - \nu \left.\frac{\partial^3 u^1}{\partial x_2^3}\right|_{x_2=0},
\end{equation*}
where, as above, $\psi=\left.G\right|_{x_2=0}$. Assuming that the third order derivative term $\nu \left.\frac{\partial^3 u^1}{\partial x_2^3}\right|_{x_2=0}$ is negligibly small, we have that $\theta$ solves the inhomogeneous heat equation
\begin{equation*}
\frac{\partial}{\partial t}\theta-2 \nu \frac{\partial^2}{\partial x_1^2} \theta = \psi.
\end{equation*}
So that the solution can be written as
\begin{equation}\label{HeatEqSol}
\theta(x_1, t) = \int_{-\infty}^{+\infty} \theta(y_1,0) h(x_1, t, y_1) \textrm{d} y_1 + \int_{0}^{t} \int_{-\infty}^{+\infty} \psi(y_1,s) h(x_1, s, y_1) \textrm{d} y_1 \textrm{d} s,
\end{equation}
with the heat kernel $h$ given as
\begin{equation}\label{HeatKernel}
h(x_1, t, y_1)=\frac{1}{(8\nu \pi t)^{1/2}} \exp \left( -\frac{|x_1-y_1|^2}{8\nu t} \right)\quad\textrm{ for }x_1, y_1 \in \mathbb{R}.
\end{equation} 

Notice that the integrals in the above formula can be written in terms of the expectations with respect to the normal random variables. Indeed, one writes
\begin{equation*}
\int_{-\infty}^{+\infty} \theta(y_1,0) h(x_1, t, y_1) \textrm{d} y_1=\mathbb{E} \left[\theta(X,0)\right],\quad\text{ where }X\sim N(x_1, 4\nu t),
\end{equation*}
and
\begin{equation*}
\int_{0}^{t} \int_{-\infty}^{+\infty} \psi(y_1,s) h(x_1, s, y_1) \textrm{d} y_1 \textrm{d} s= \int_{0}^{t} \mathbb{E} \left[\psi(Y^s,s)\right] \textrm{d} s,\quad\text{ where }Y^s\sim N(x_1, 4\nu s).
\end{equation*}
This representation allows for the Monte-Carlo approximation of the solution \eqref{HeatEqSol} which gives 
\begin{equation}\label{ThetaApprox}
\theta(x_1,t_{k+1})\approx \frac{1}{N} \sum_{i=1}^{N} \theta(X_i,0) + \frac{h}{N} \sum_{i=1}^{N} \sum_{j=0}^{k} \psi(Y_i^{t_j},t_j),
\end{equation}
with $X_i$ and $Y_i^{t_j}$ drawn independently from $N(x_1, 4\nu t)$ and $N(x_1, 4\nu t_j)$, respectively.

Notice that the expressions for the boundary vorticity \eqref{ThetaApprox} and the velocity \eqref{NS1U} contain time-dependent summations. As in \citep{CherepanovQian2023}, we store the results of these summations for each index $i$ and $(i_1, i_2)$ respectively, which allows us to update the sum by computing one term per index at each time step. However, since we have indicators with the last boundary crossing times $\gamma_{t_k}(X^{i_1, i_2})$ in \eqref{NS1U}, we also keep track of the crossings for each $(i_1, i_2)$. We set the corresponding sum to zero at each step when the crossing happens, and after doing so, we continue updating the sum as before.

\textit{Experiment 1.} In this experiment, we assume that the initial velocity is of the form $u(x_1,x_2,0) = (-U_0 1_{\left\{ x_2 > 0 \right\}}, 0)$, i.e. a constant horizontal field formally satisfying the no-slip condition. This means that the vorticity is initialised as $\omega_0(x_1, x_2) = U_0 1_{\left\{ x_2 =0 \right\}}$, and the motion is affected by the external force $G = G_0 1_{\left\{ x_2 = 0 \right\}}$, which is concentrated at the boundary as well. 

For this simulation, we consider the half-plane domain within the limits of the box $-H \leq x_1 \leq H$, $0 \leq x_2 \leq H$ with $H = 6$ for the size of the domain. We also choose $H_0 = 0.1$ for the boundary layer thickness. The lattice points are given by \eqref{PlateBLayerLattice} and \eqref{PlateOLayerLattice} with $N_0=30, N_1=30, N_2=45$. Therefore, the mesh sizes are given by $h_0=\frac{H}{N_0} = 0.4, h_1=\frac{H}{N_1} = 0.2, h_2=\frac{H_0}{N_2} \approx 0.0022$. The parameter $\varepsilon$ in \eqref{NS1U} is taken to be $0.02$.

We conduct this experiment with the velocity and force constants $U_0 = 0.01, G_0 = -1.0$, and the viscosity $\nu = 0.1$. The simulation is conducted with time steps $t=0.01$ with the scheme described above, i.e. with boundary vorticity given by \eqref{ThetaApprox}, the results are presented in Figures \ref{Exp1aFigOFlow} and \ref{Exp1aFigBFlow}. In these and below plots, the streamlines are coloured by the magnitude of the velocity, and the background colour represents the vorticity value. We also use the scheme from \citep{CherepanovQian2023}, i.e. computing the vorticity as the derivative of \eqref{NS1U}, and provide the results in Figures \ref{Exp1bFigOFlow} and \ref{Exp1bFigBFlow}. The boundary vorticity $\theta$ and the third derivative term $\nu \left.\frac{\partial^3 u^1}{\partial x_2^3}\right|_{x_2=0}$ are plotted in \ref{Exp1bFigBStress} as functions of $x_1$ and $t$.

Comparing these simulations, we conclude that the scheme with the boundary vorticity $\theta$ approximated from the dynamical equation gives results similar to those when it is computed explicitly as the derivative of the velocity field. The boundary vorticity seems to affect mostly the boundary flow (as it is present in the summation over the boundary lattice), though we also observe that the produced flows exhibit similar behaviour close to the boundary. 

\textit{Experiment 2.} We initialise the vorticity as $\omega_0(x_1,x_2)=-U_0 \left(3 - \frac{x_1}{H}\right) \left(1 - x_2\right) 1_{\{0\leq x_2 \leq 1\}}$, which yields non-trivial boundary vorticity $\theta(x_1, 0)=-U_0 \left(3 - \frac{x_1}{H}\right)$ with $U_0=3.0$. The external force $G$ is taken to be identically zero. We use the same lattice points and parameters as above, but choose time steps $t=0.03$. The simulation is conducted with the vorticity computed as in \eqref{ThetaApprox}, the results of this simulation are given in Figures \eqref{Exp2FigOFlow} and \eqref{Exp2FigBFlow}.

\begin{figure}
    \centering
    \subfloat[\centering $t=0.05$]{{\includegraphics[width=.5\linewidth]{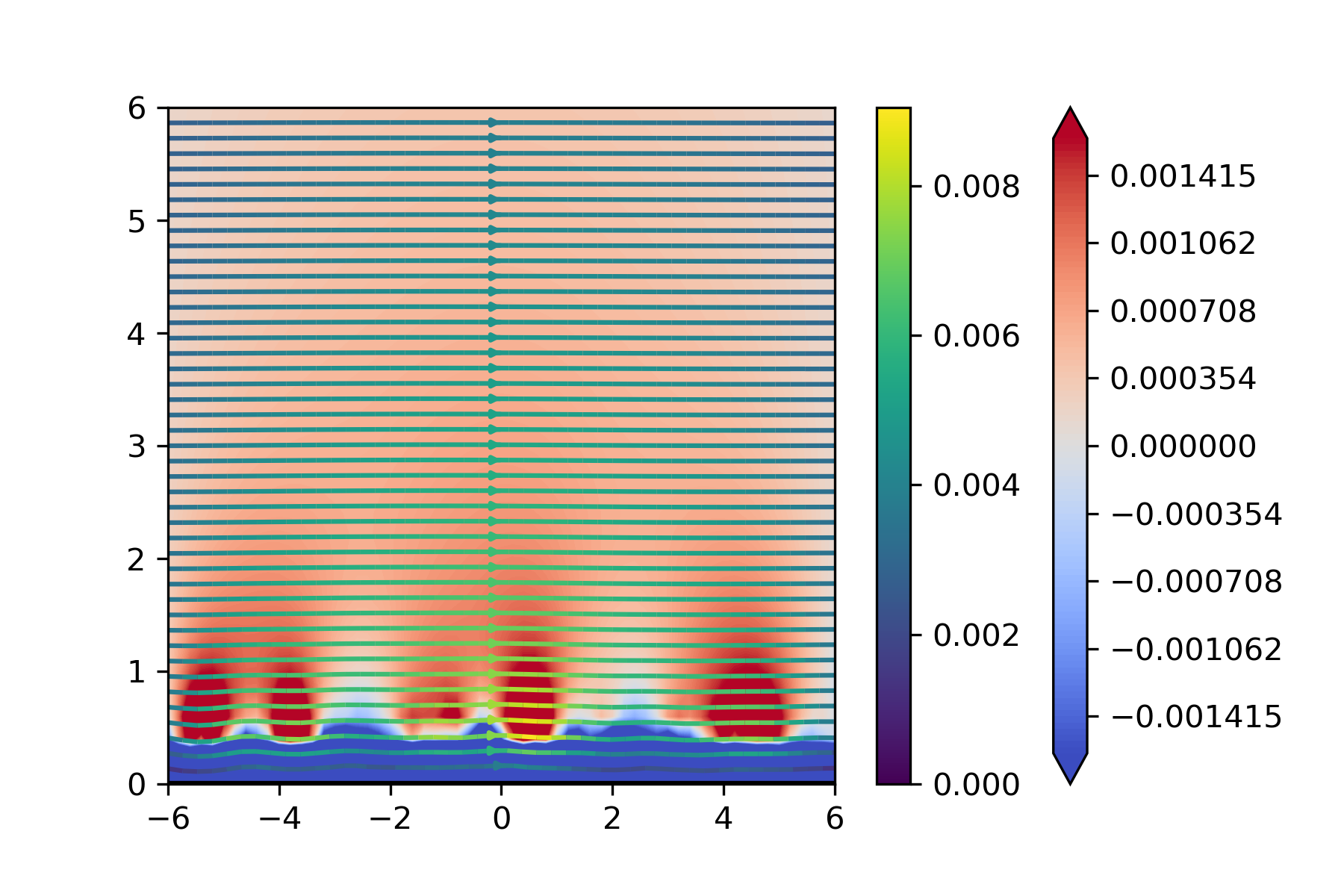} }}
    \subfloat[\centering $t=0.25$]{{\includegraphics[width=.5\linewidth]{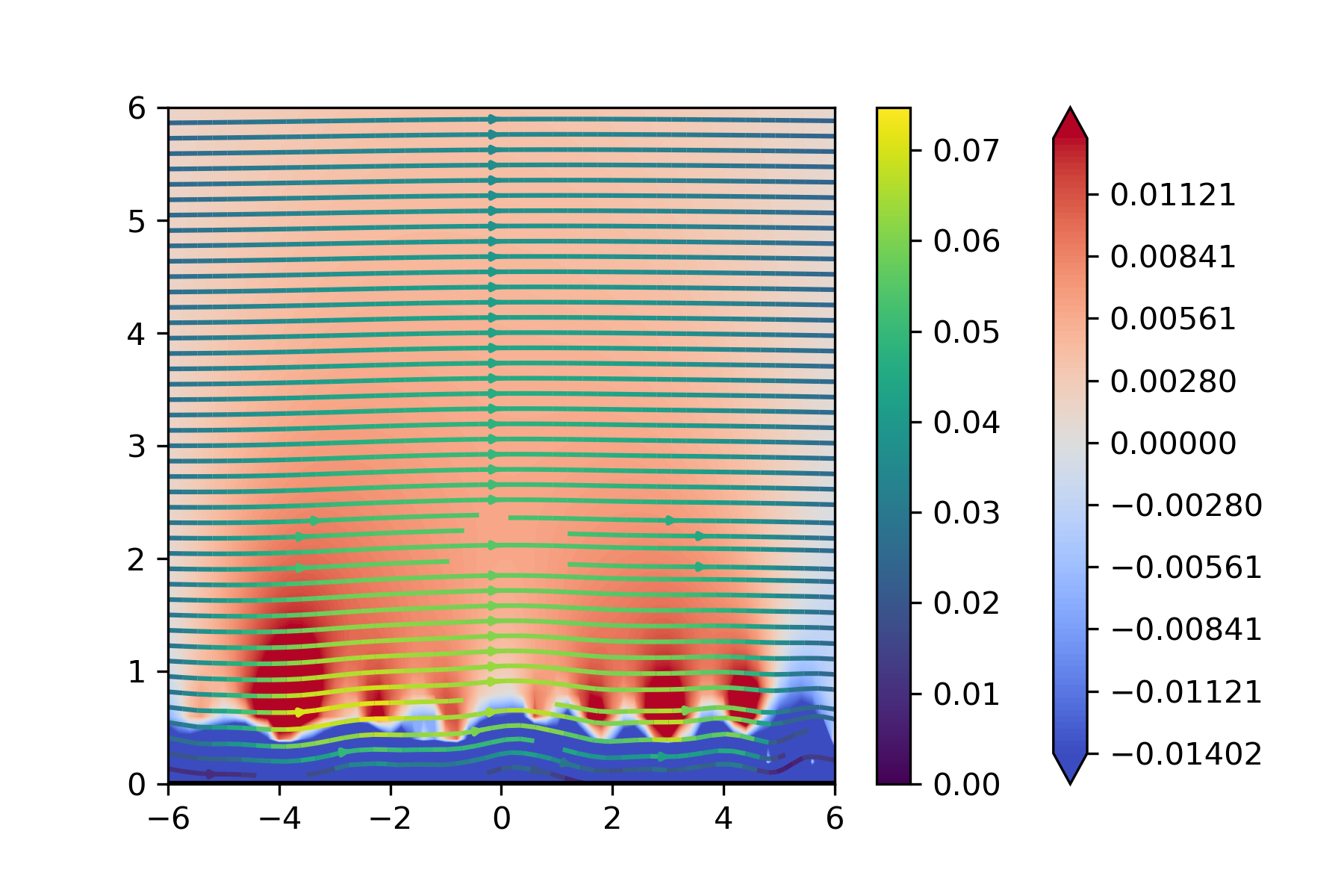}}}
    \qquad
    \subfloat[\centering $t=0.5$]{{\includegraphics[width=.5\linewidth]{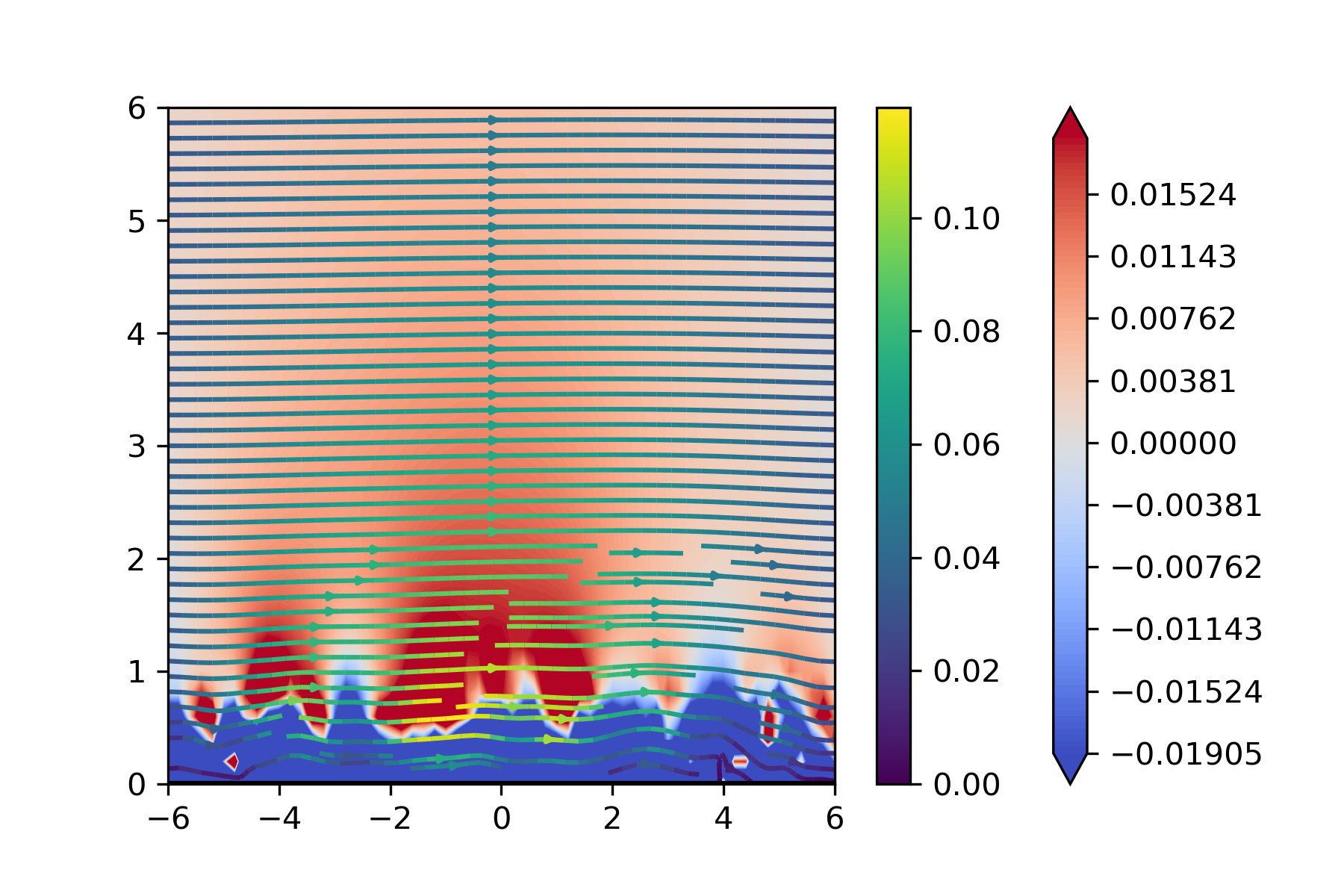} }}
    \subfloat[\centering $t=1.0$]{{\includegraphics[width=.5\linewidth]{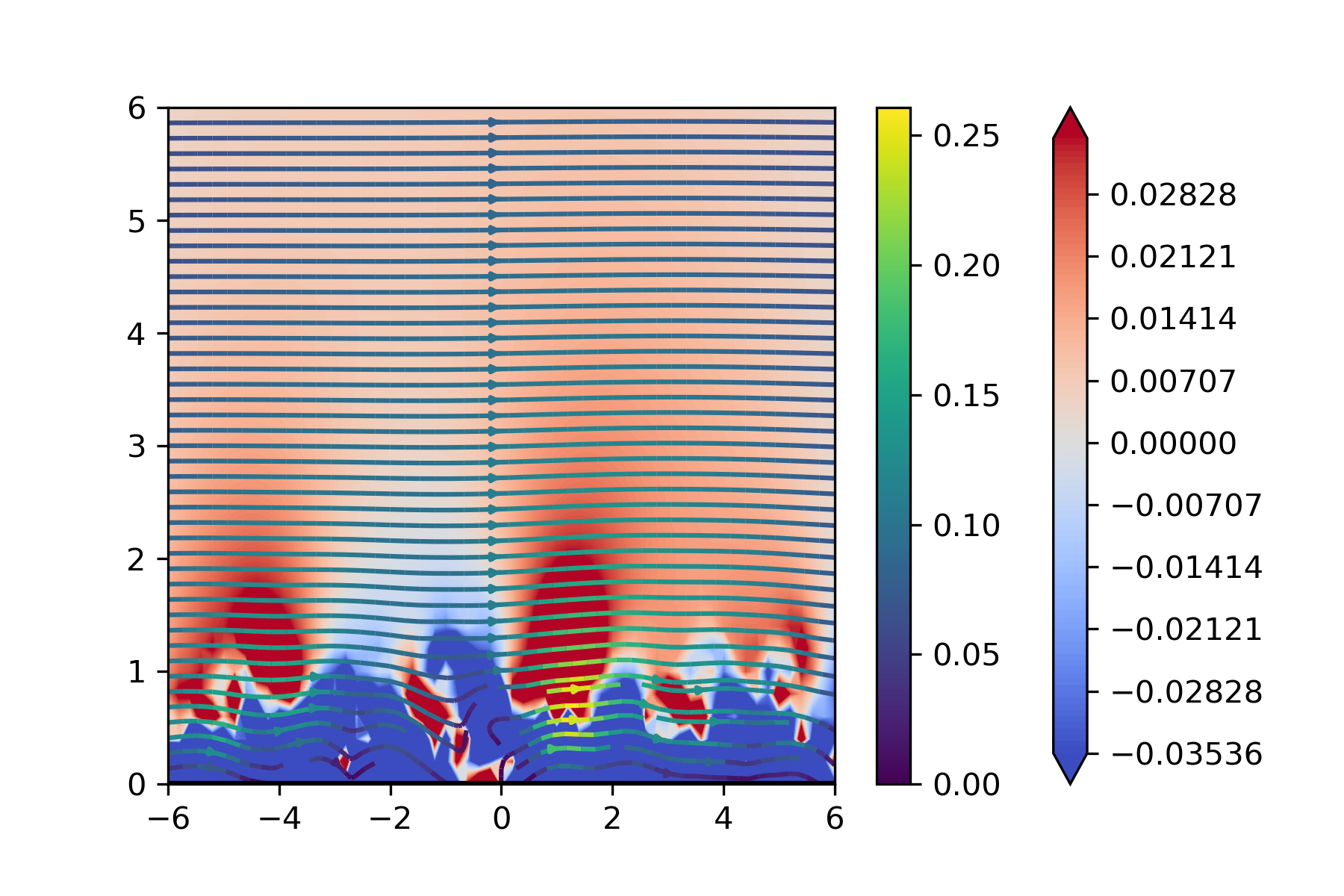}}}
    \caption{The outer layer flow at different times $t$.}
    \label{Exp1aFigOFlow}
\end{figure}

\begin{figure}
    \centering
    \subfloat[\centering $t=0.05$]{{\includegraphics[width=.5\linewidth]{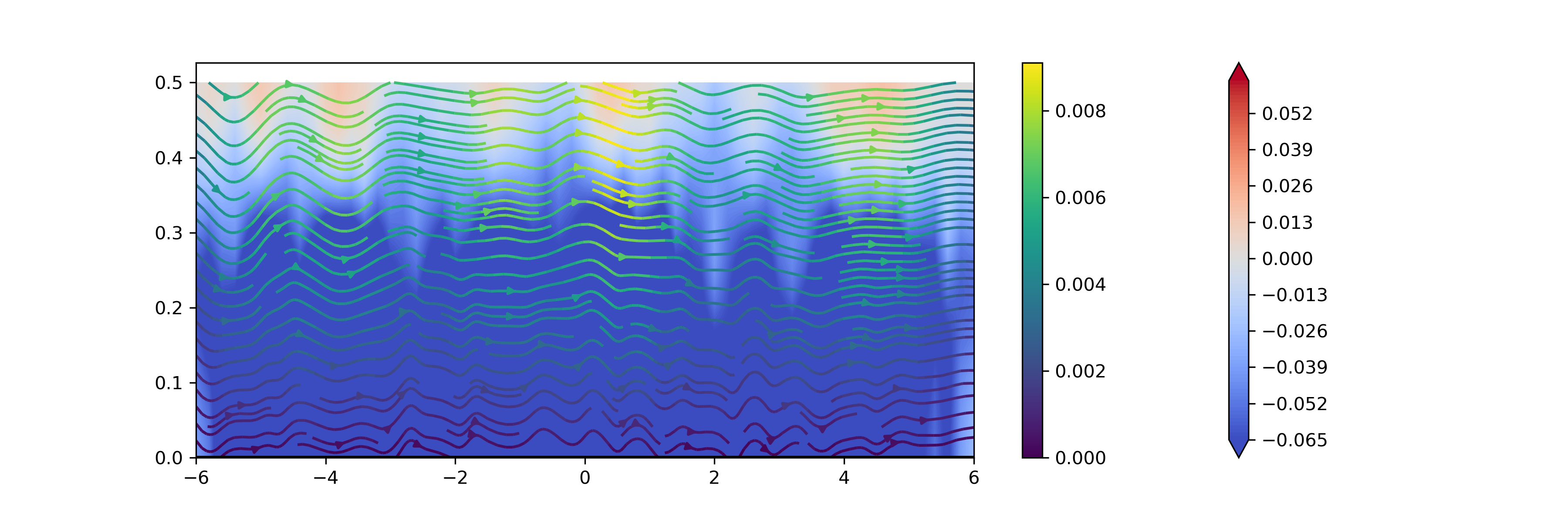} }}
    \subfloat[\centering $t=0.25$]{{\includegraphics[width=.5\linewidth]{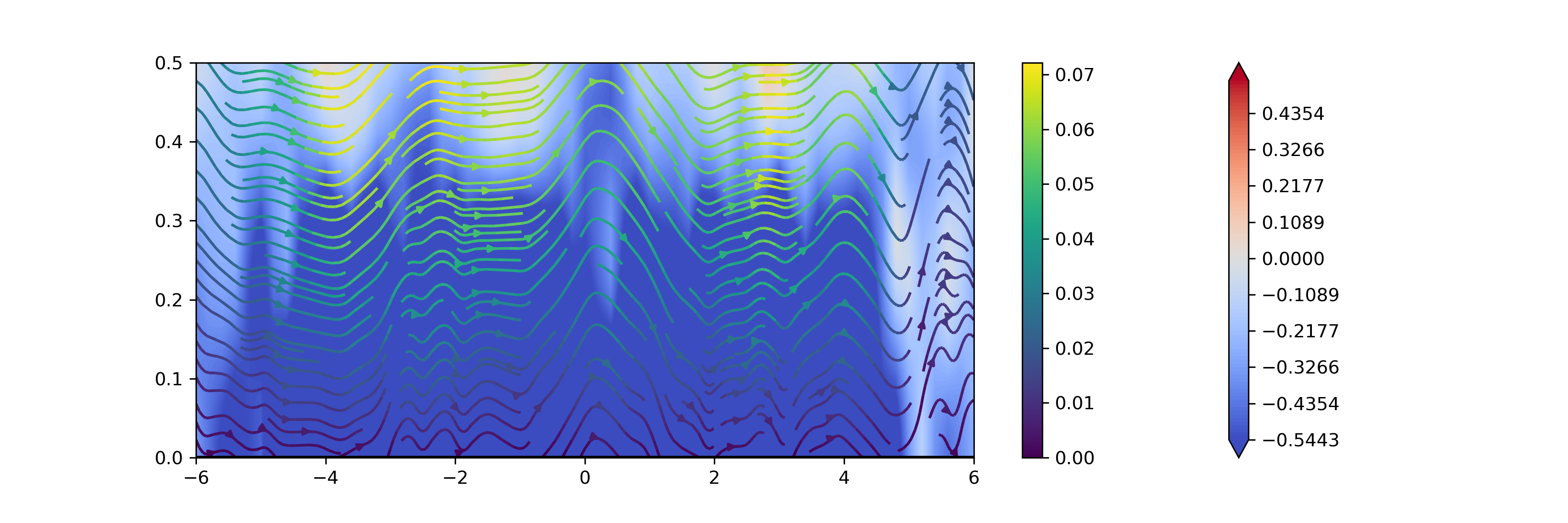}}}
    \qquad
    \subfloat[\centering $t=0.5$]{{\includegraphics[width=.5\linewidth]{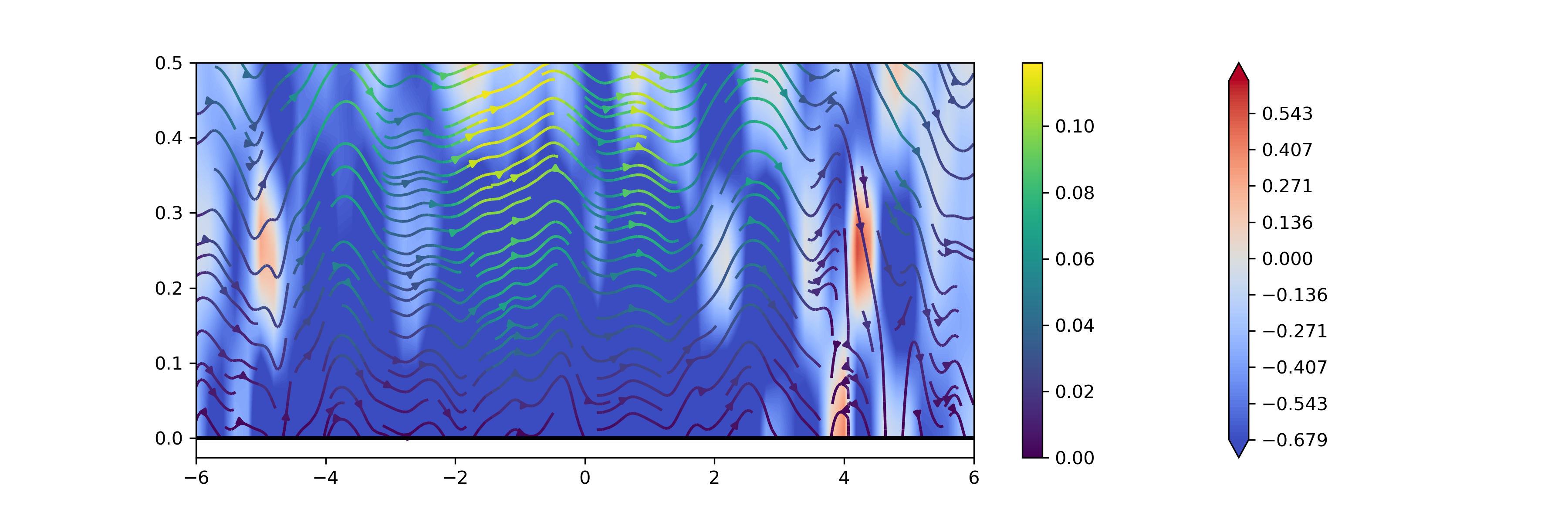} }}
    \subfloat[\centering $t=1.0$]{{\includegraphics[width=.5\linewidth]{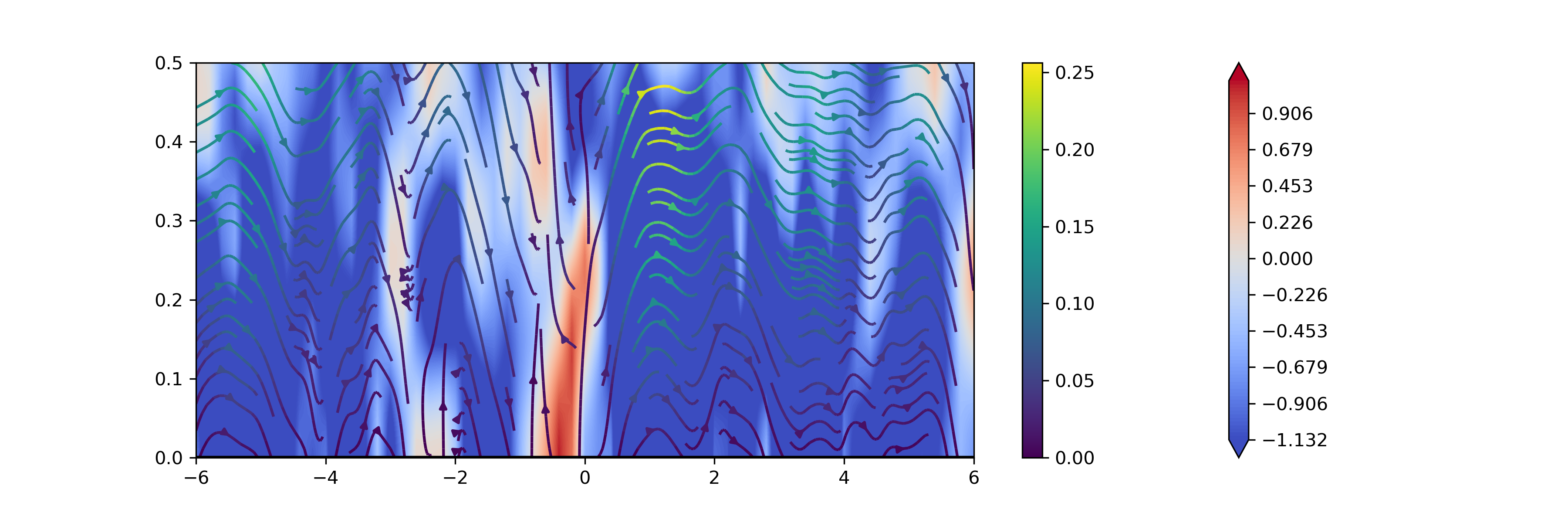}}}
    \caption{The boundary layer flow at different times $t$.}
    \label{Exp1aFigBFlow}
\end{figure}

\begin{figure}
    \centering
    \subfloat[\centering $t=0.05$]{{\includegraphics[width=.5\linewidth]{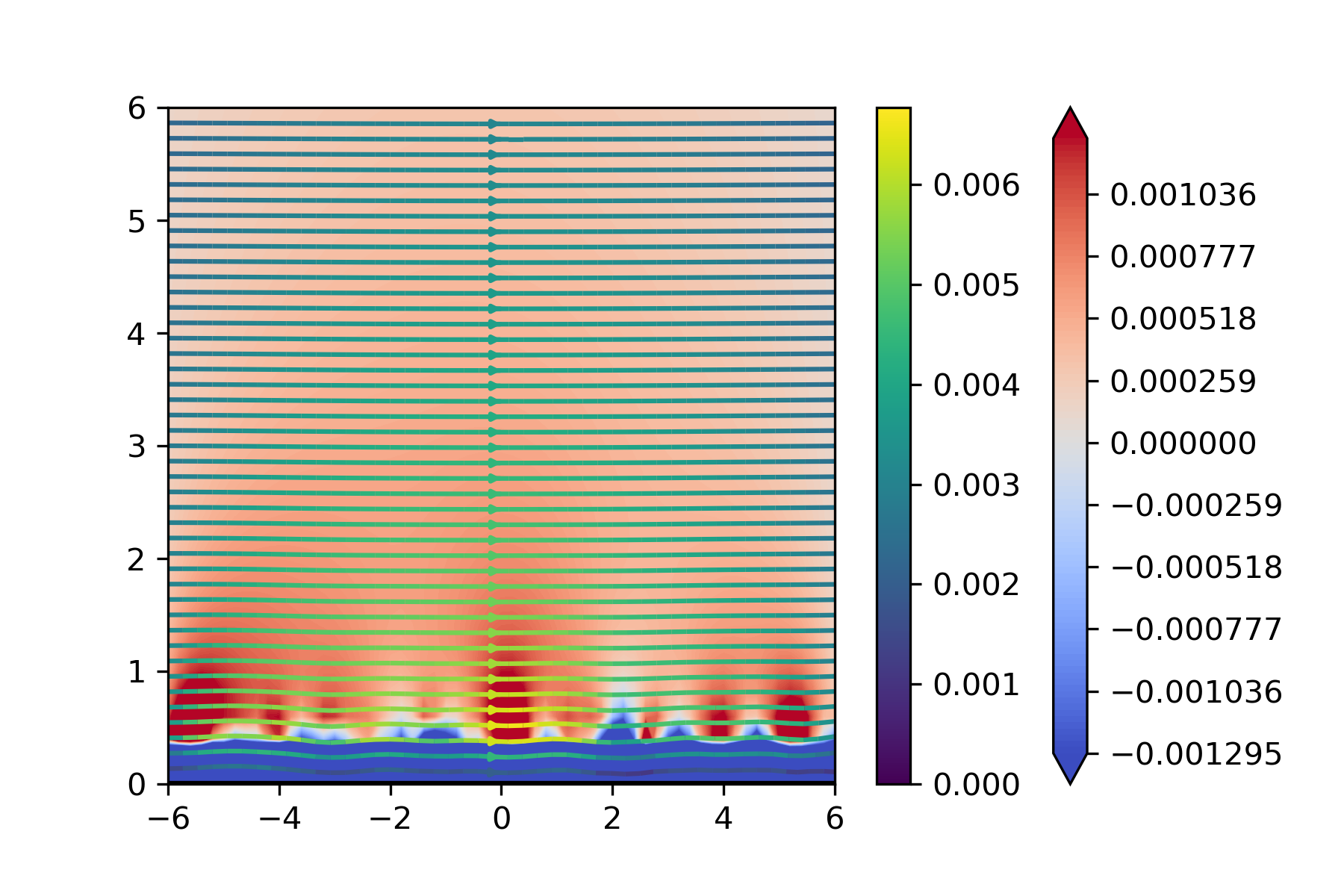} }}
    \subfloat[\centering $t=0.25$]{{\includegraphics[width=.5\linewidth]{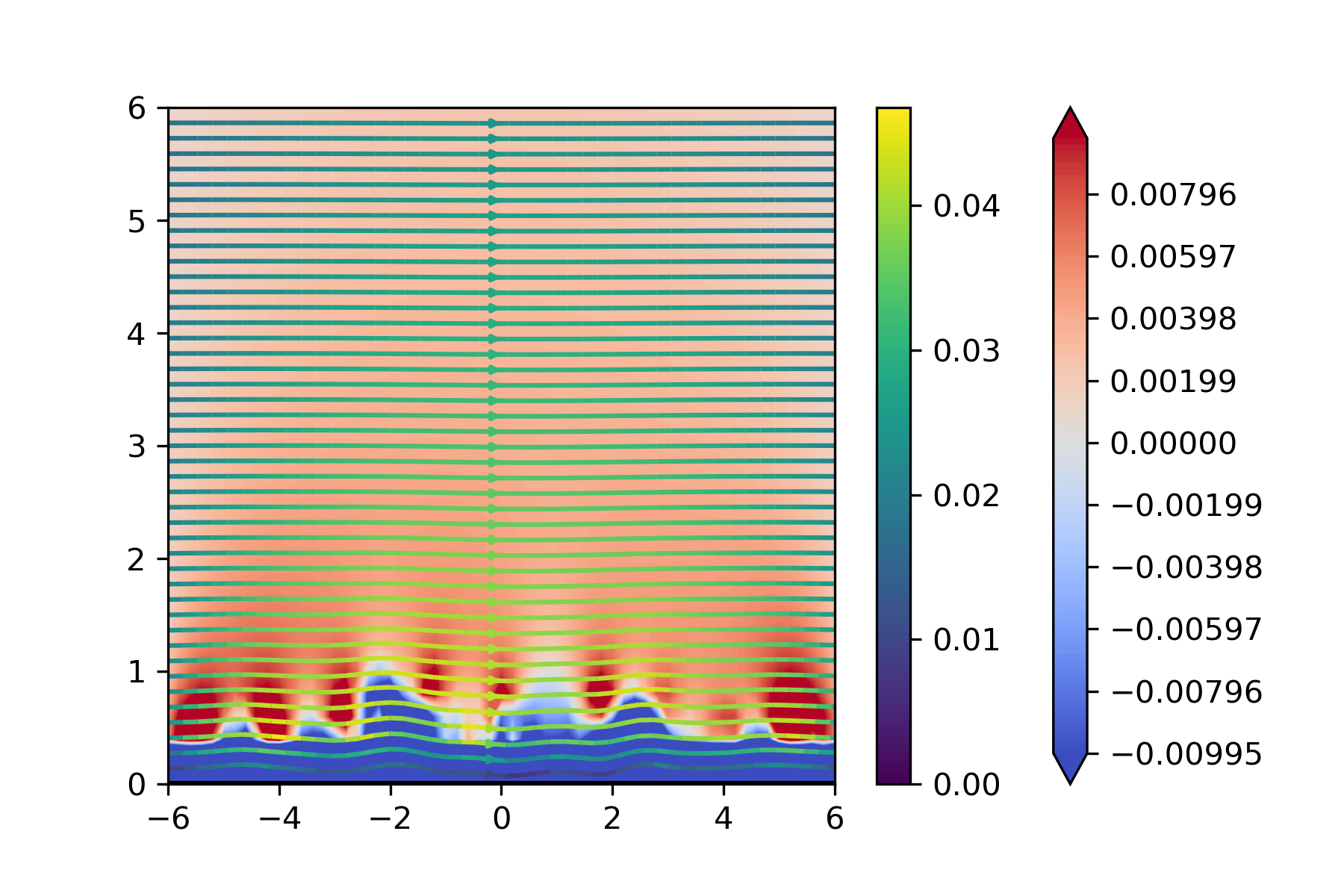}}}
    \qquad
    \subfloat[\centering $t=0.5$]{{\includegraphics[width=.5\linewidth]{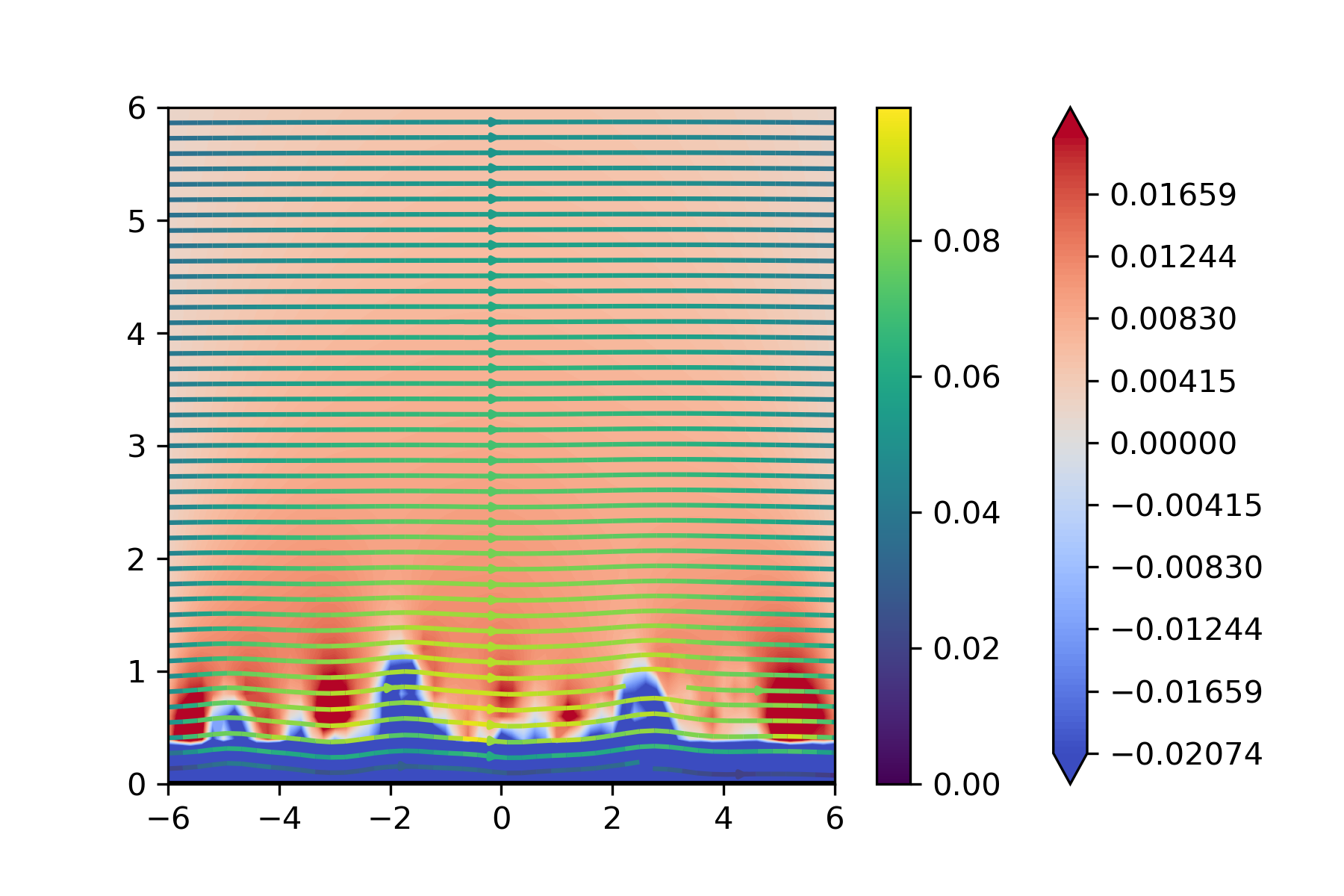} }}
    \subfloat[\centering $t=1.0$]{{\includegraphics[width=.5\linewidth]{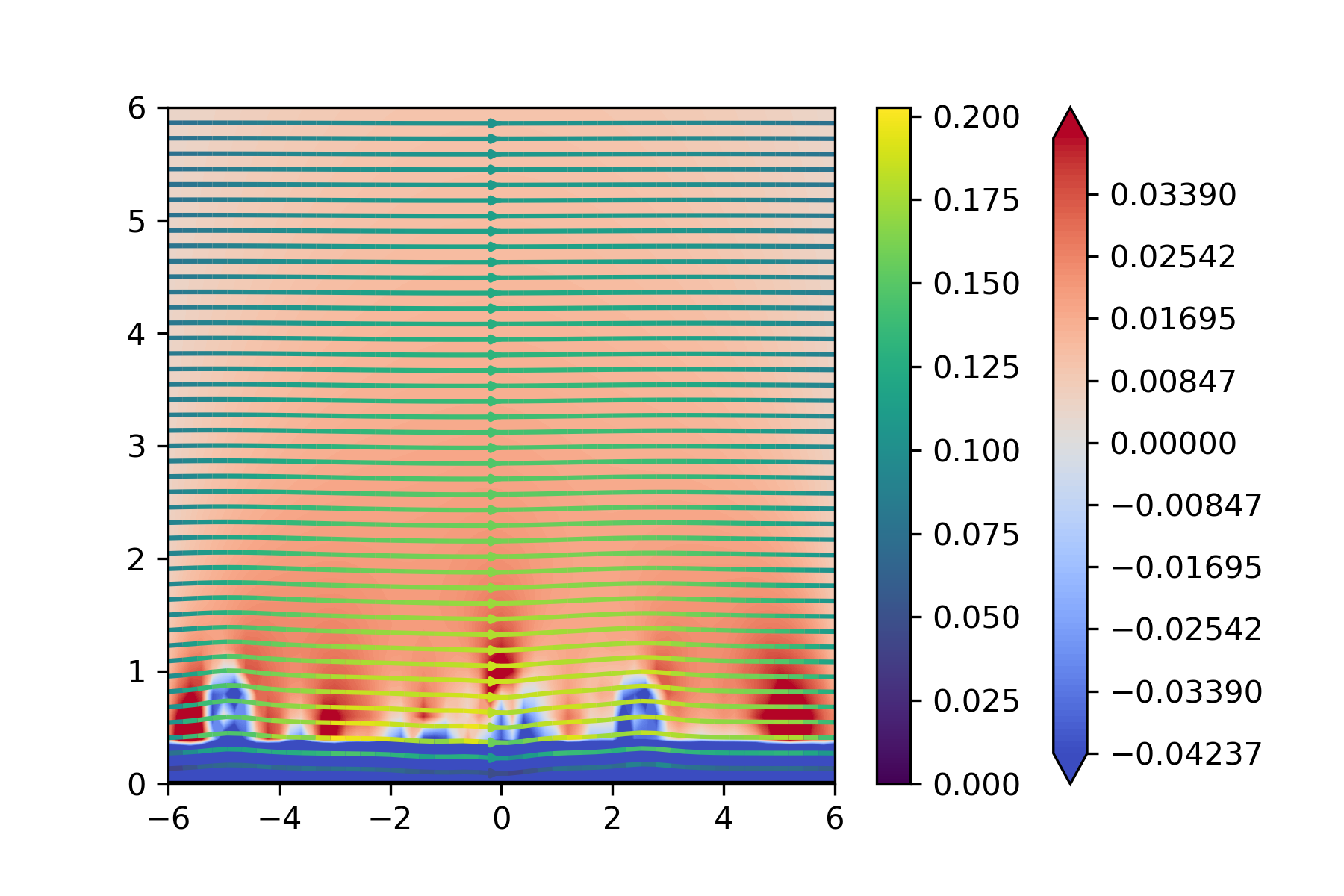}}}
    \caption{The outer layer flow at different times $t$.}
    \label{Exp1bFigOFlow}
\end{figure}

\begin{figure}
    \centering
    \subfloat[\centering $t=0.05$]{{\includegraphics[width=.5\linewidth]{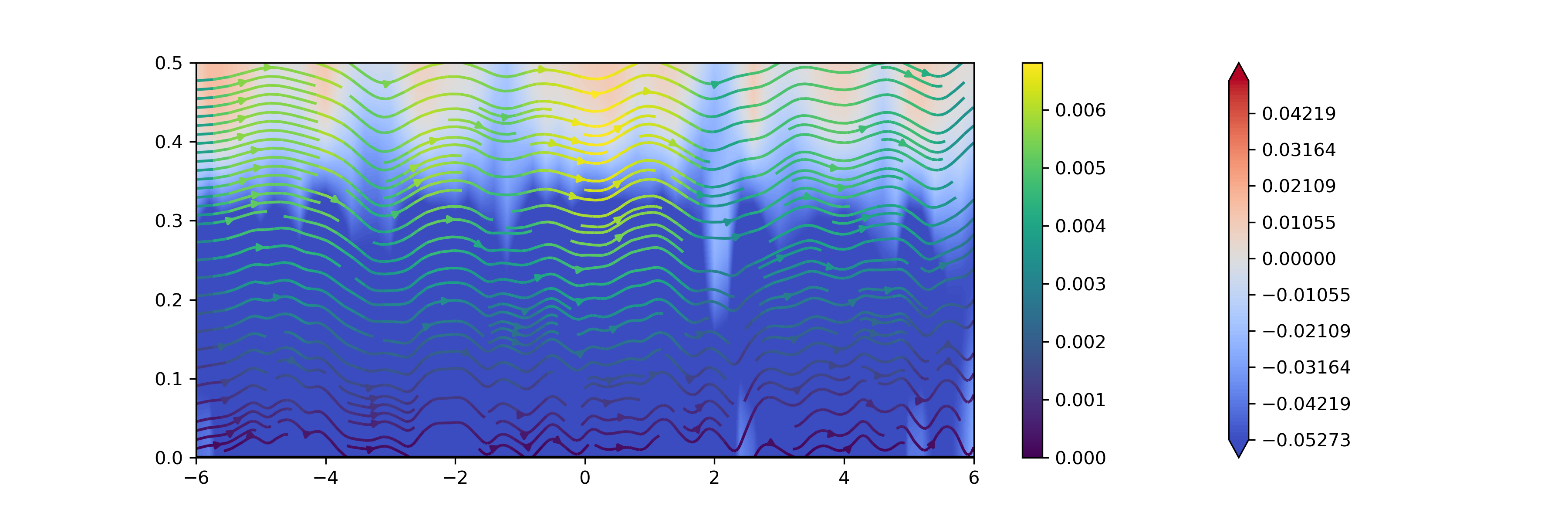} }}
    \subfloat[\centering $t=0.25$]{{\includegraphics[width=.5\linewidth]{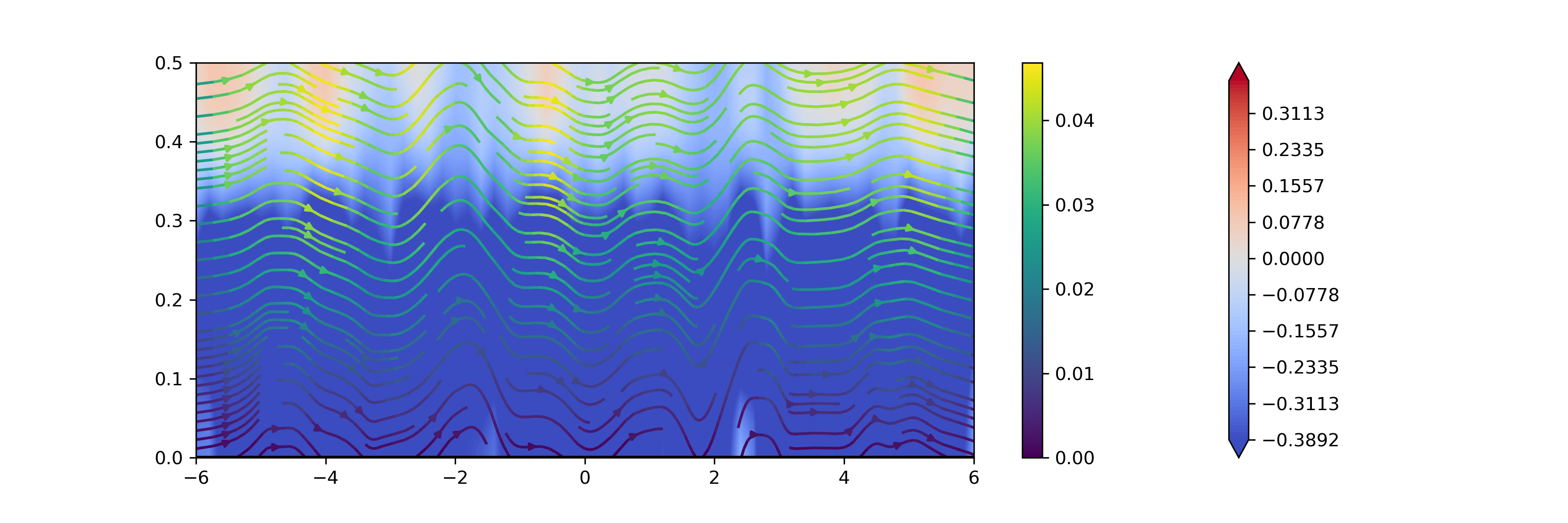}}}
    \qquad
    \subfloat[\centering $t=0.5$]{{\includegraphics[width=.5\linewidth]{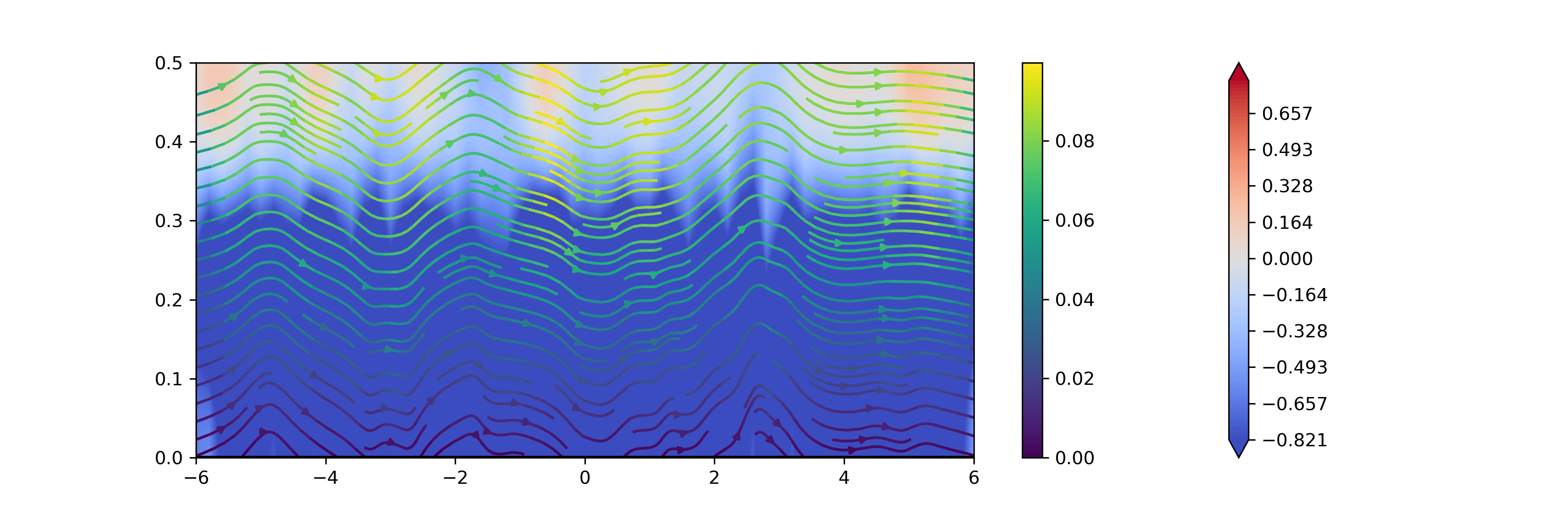} }}
    \subfloat[\centering $t=1.0$]{{\includegraphics[width=.5\linewidth]{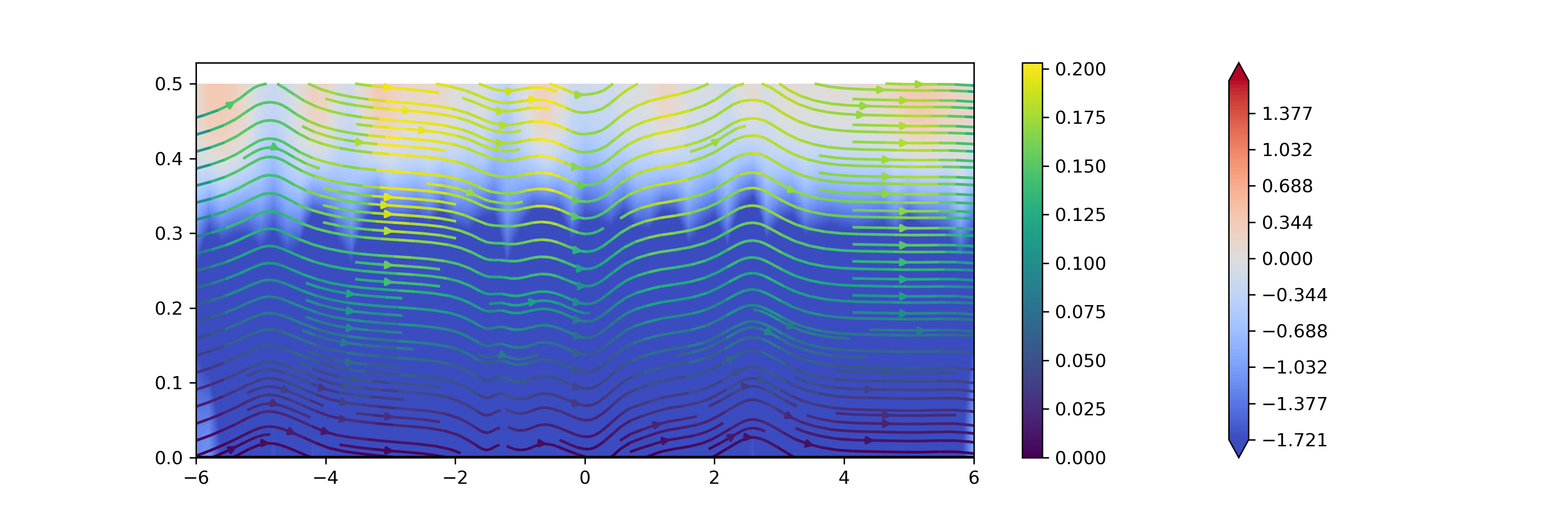}}}
    \caption{The boundary layer flow at different times $t$.}
    \label{Exp1bFigBFlow}
\end{figure}

\begin{figure}
    \centering
    \subfloat[\centering \text{Boundary vorticity} $\theta(x_1,t)$]{{\includegraphics[width=.5\linewidth]{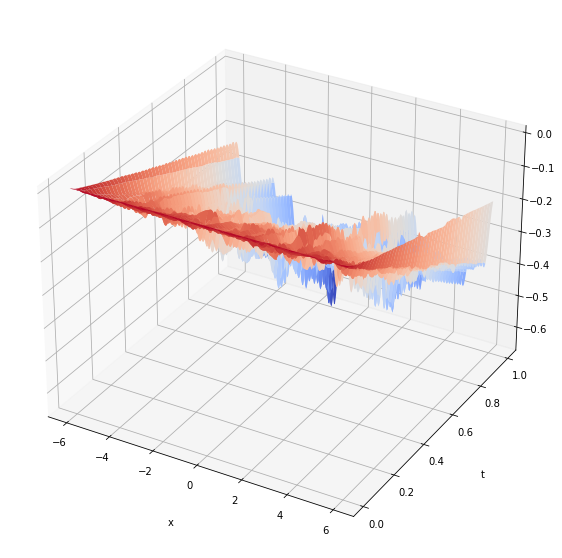}}}
    \subfloat[\centering \text{Third order derivative term} $\nu \frac{\partial^3 u^1}{\partial x_2^3} (x_1,t)$]{{\includegraphics[width=.5\linewidth]{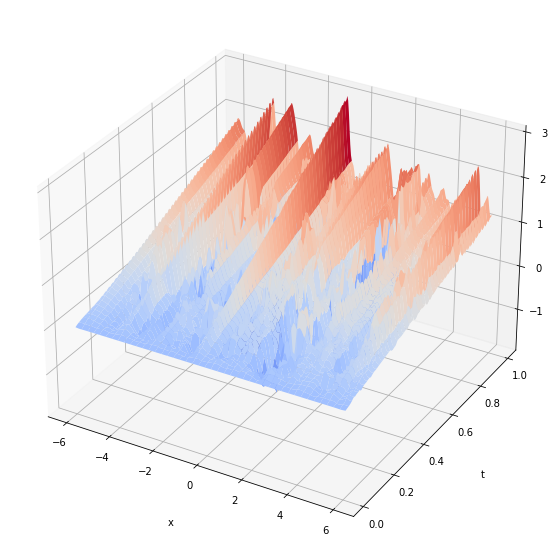} }}
    \caption{The boundary vorticity $\theta$ and the $\nu \frac{\partial^3 u^1}{\partial x_2^3}$ as functions of position at boundary $x_1$ and time $t$.}
    \label{Exp1bFigBStress}
\end{figure}

\begin{figure}
    \centering
    \subfloat[\centering $t=0.15$]{{\includegraphics[width=.5\linewidth]{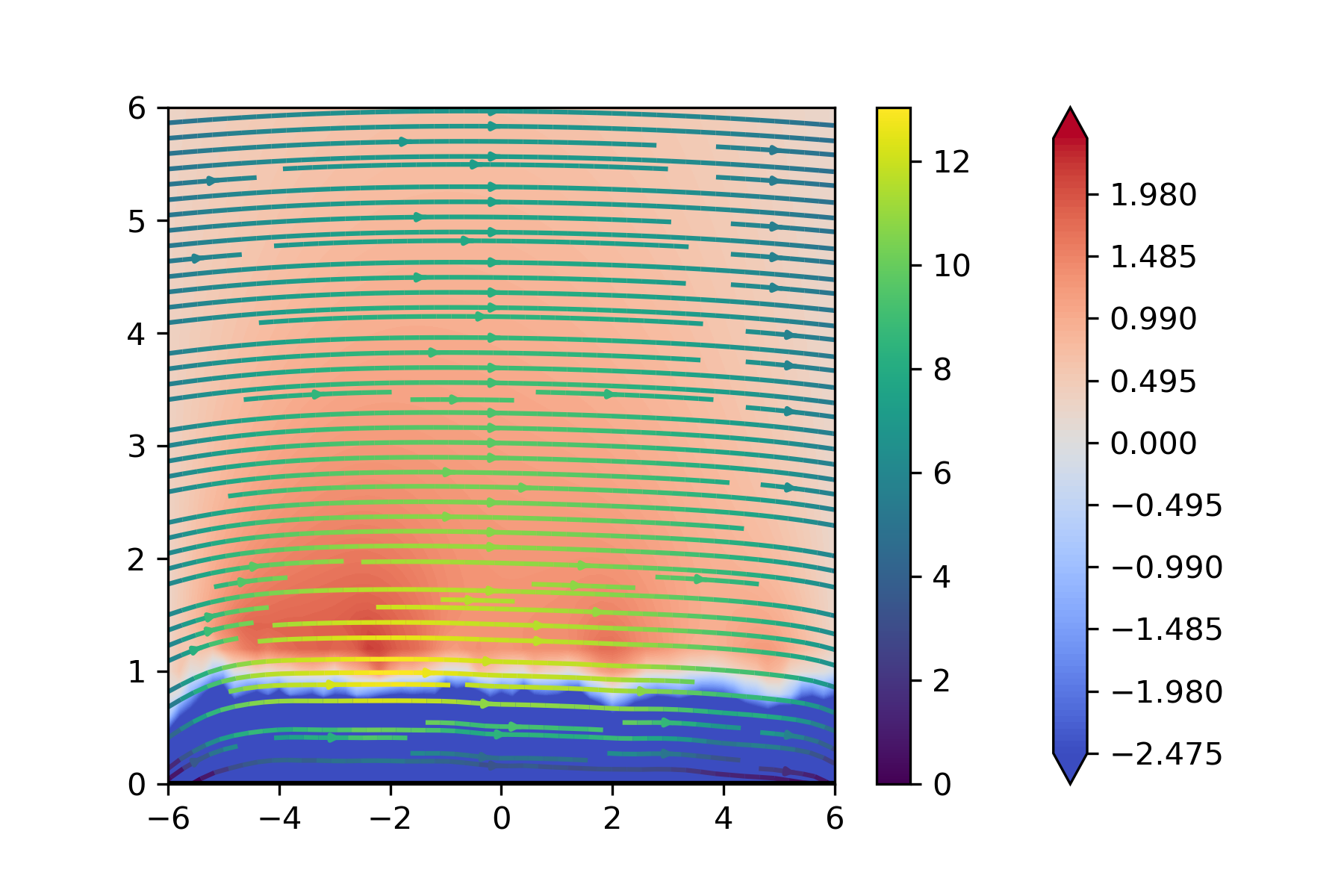} }}
    \subfloat[\centering $t=0.75$]{{\includegraphics[width=.5\linewidth]{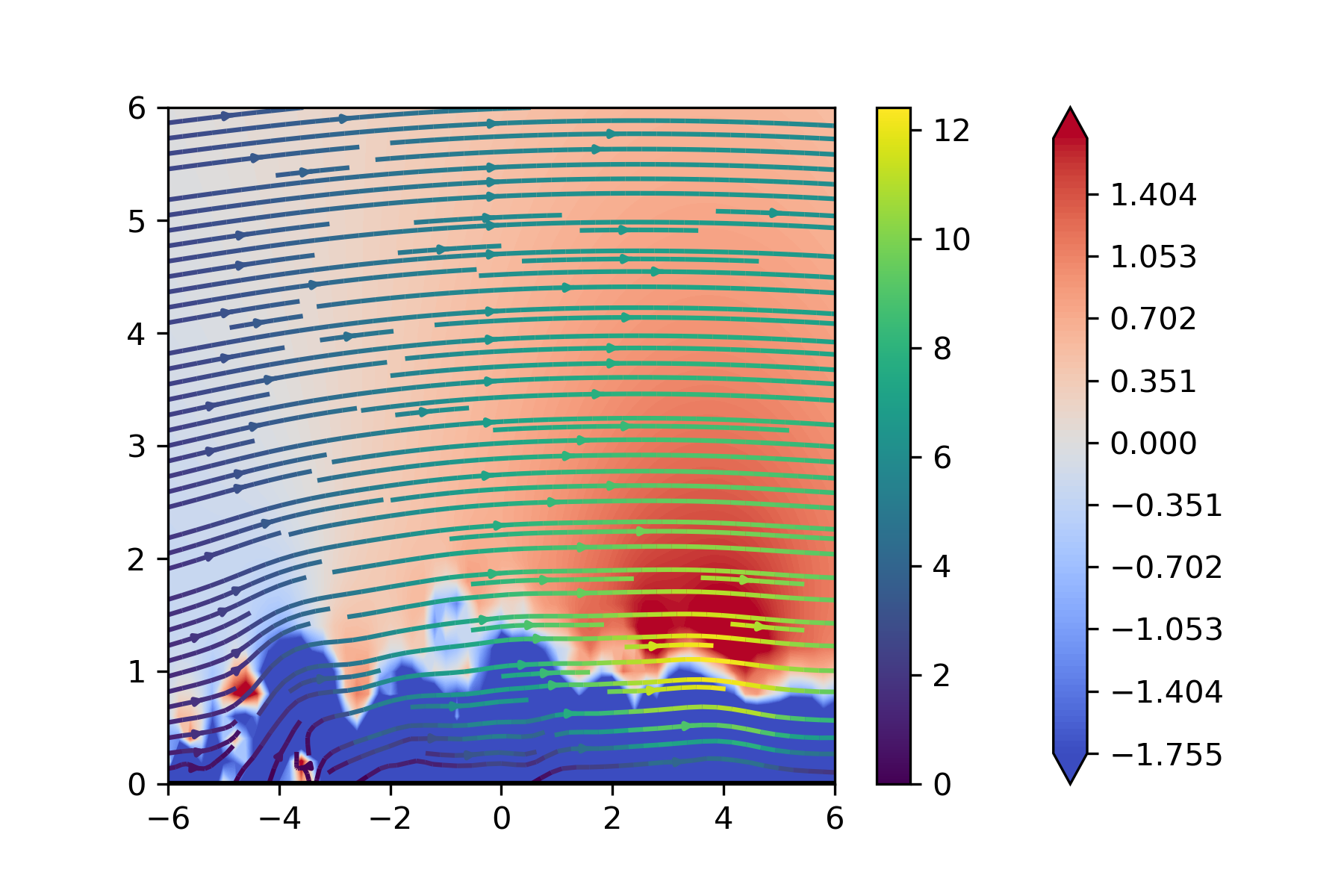}}}
    \qquad
    \subfloat[\centering $t=1.5$]{{\includegraphics[width=.5\linewidth]{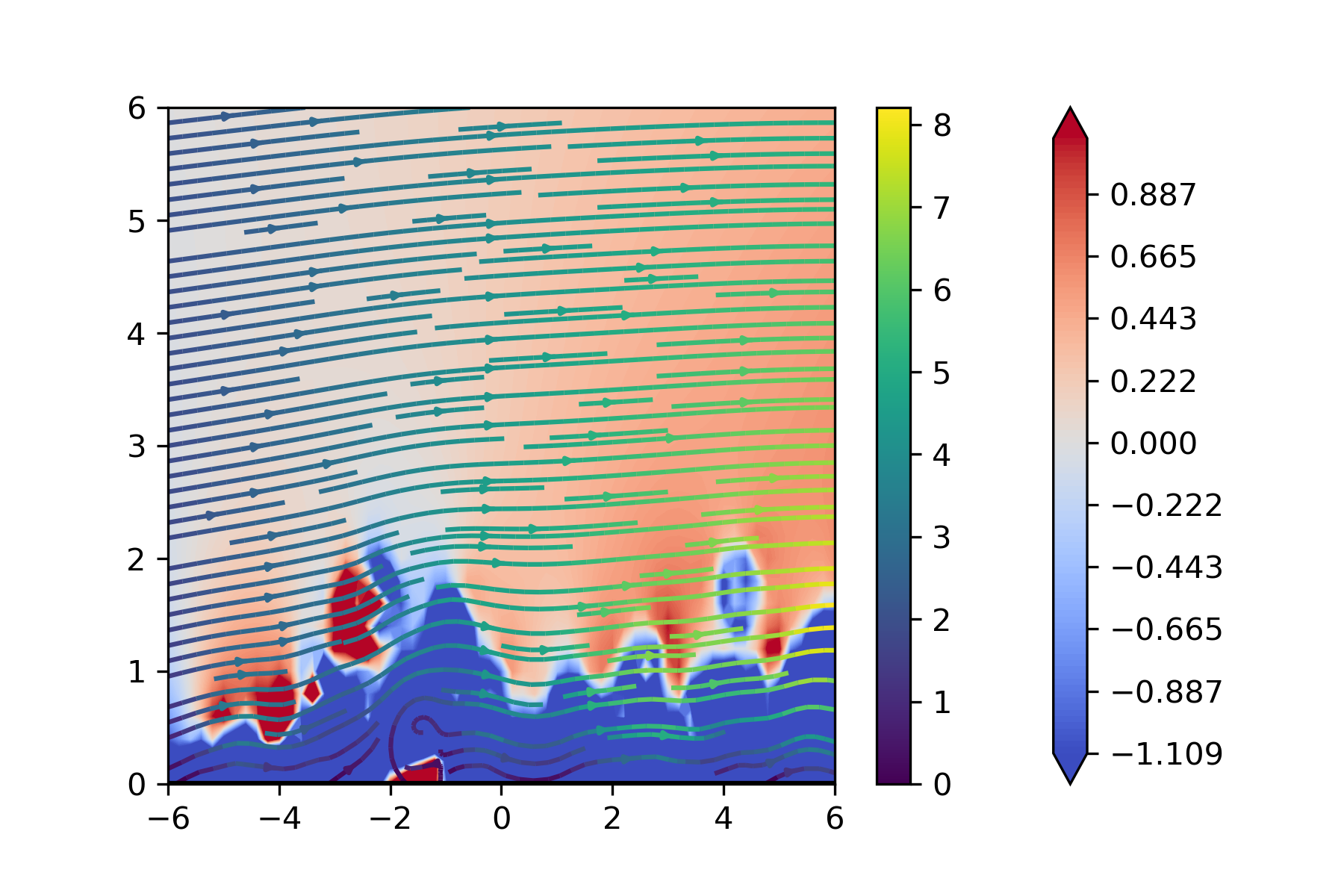} }}
    \subfloat[\centering $t=3.0$]{{\includegraphics[width=.5\linewidth]{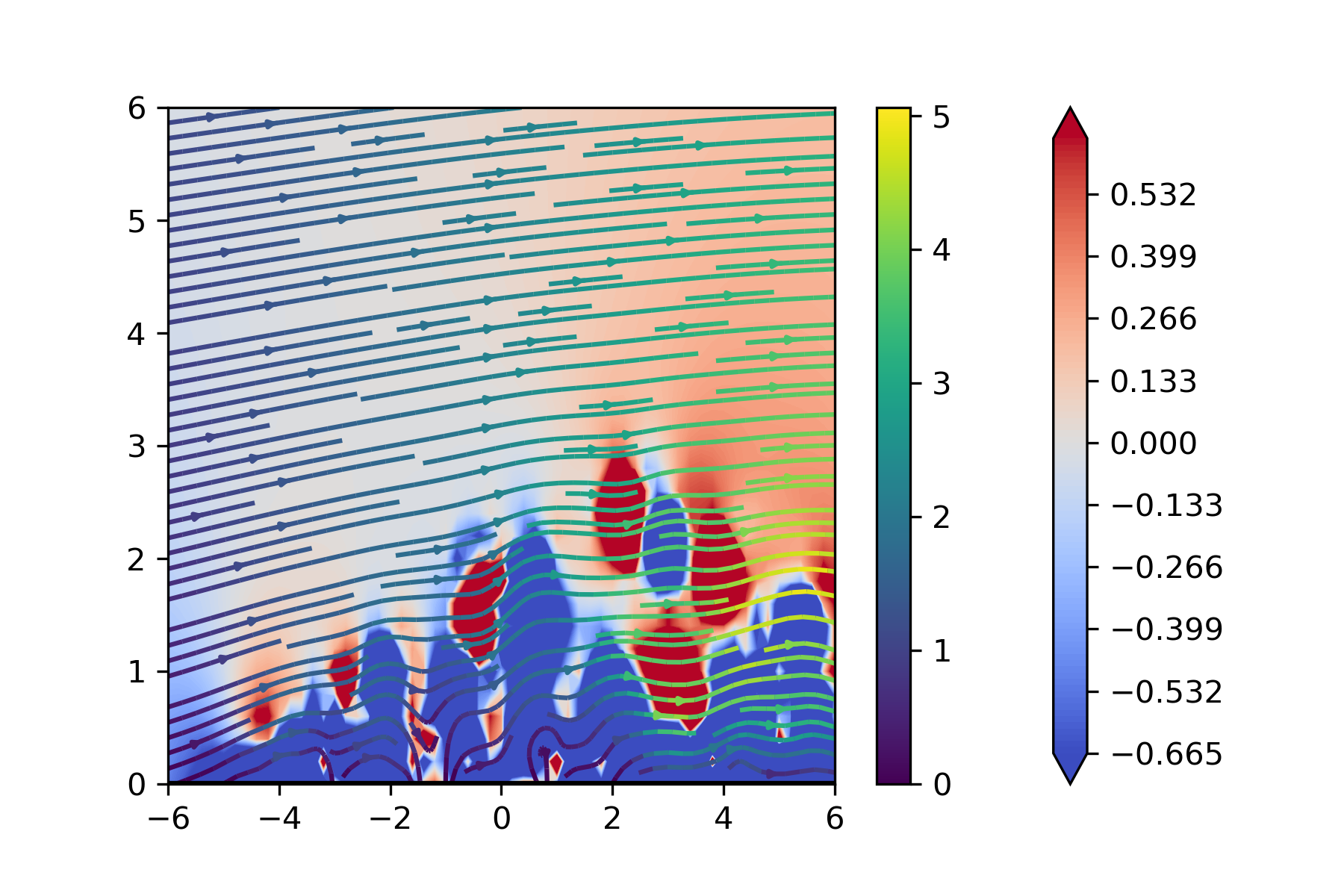}}}
    \caption{The outer layer flow at different times $t$.}
    \label{Exp2FigOFlow}
\end{figure}

\begin{figure}
    \centering
    \subfloat[\centering $t=0.15$]{{\includegraphics[width=.5\linewidth]{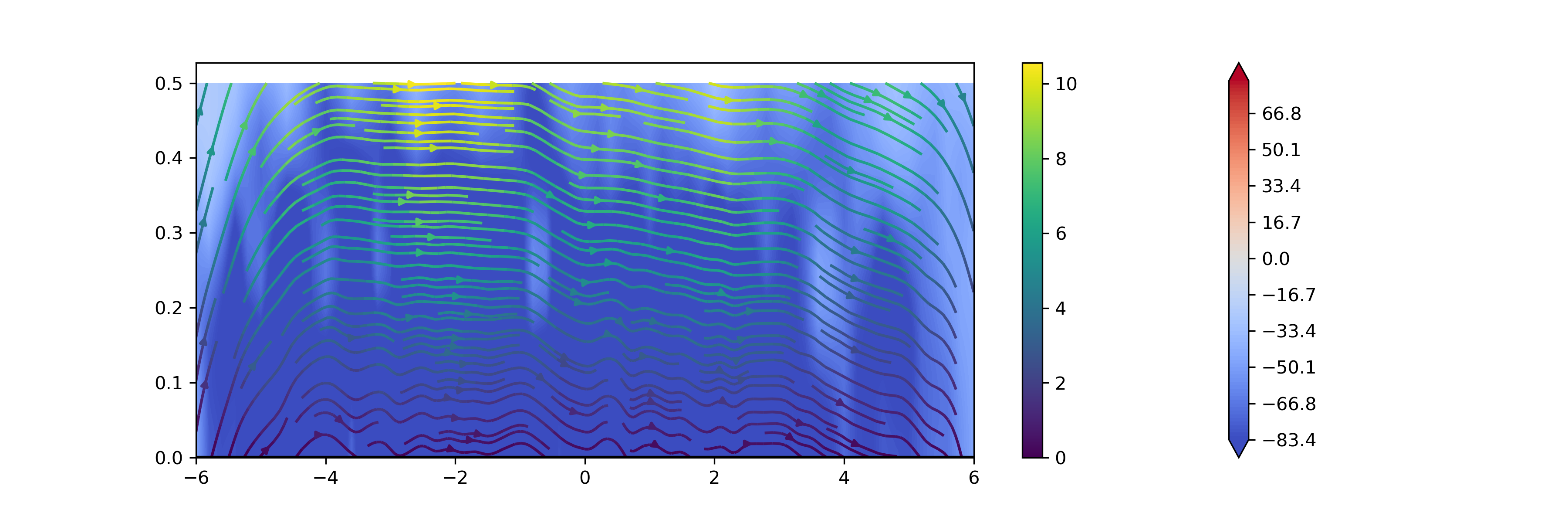} }}
    \subfloat[\centering $t=0.75$]{{\includegraphics[width=.5\linewidth]{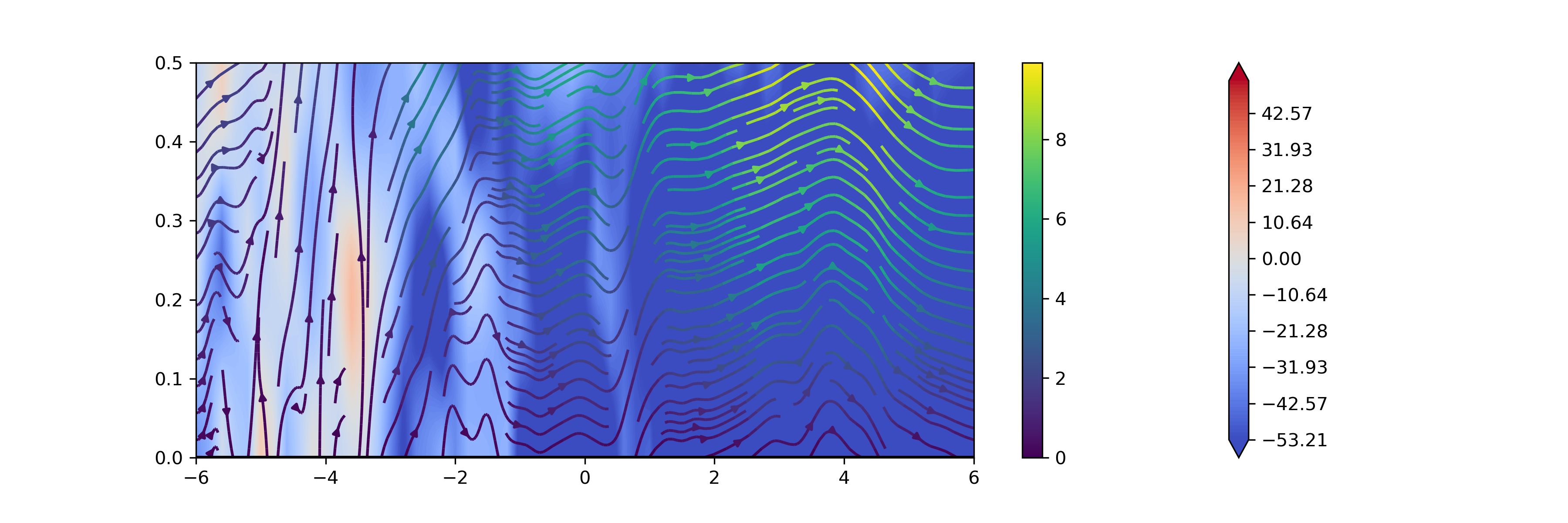}}}
    \qquad
    \subfloat[\centering $t=1.5$]{{\includegraphics[width=.5\linewidth]{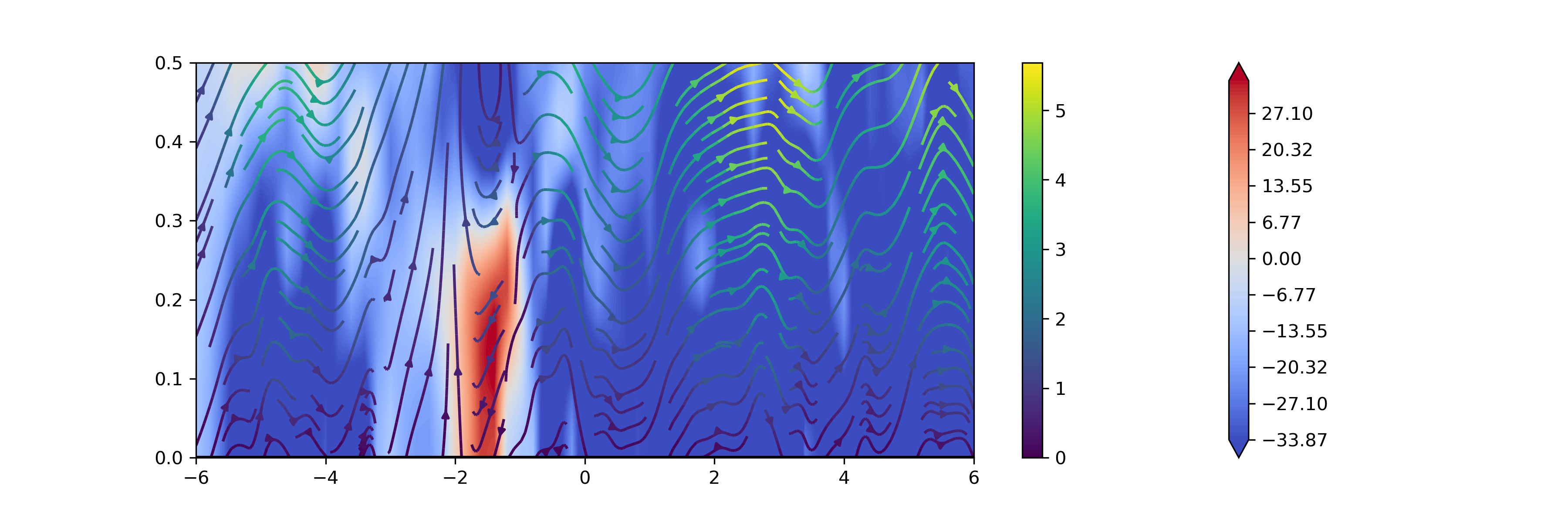} }}
    \subfloat[\centering $t=3.0$]{{\includegraphics[width=.5\linewidth]{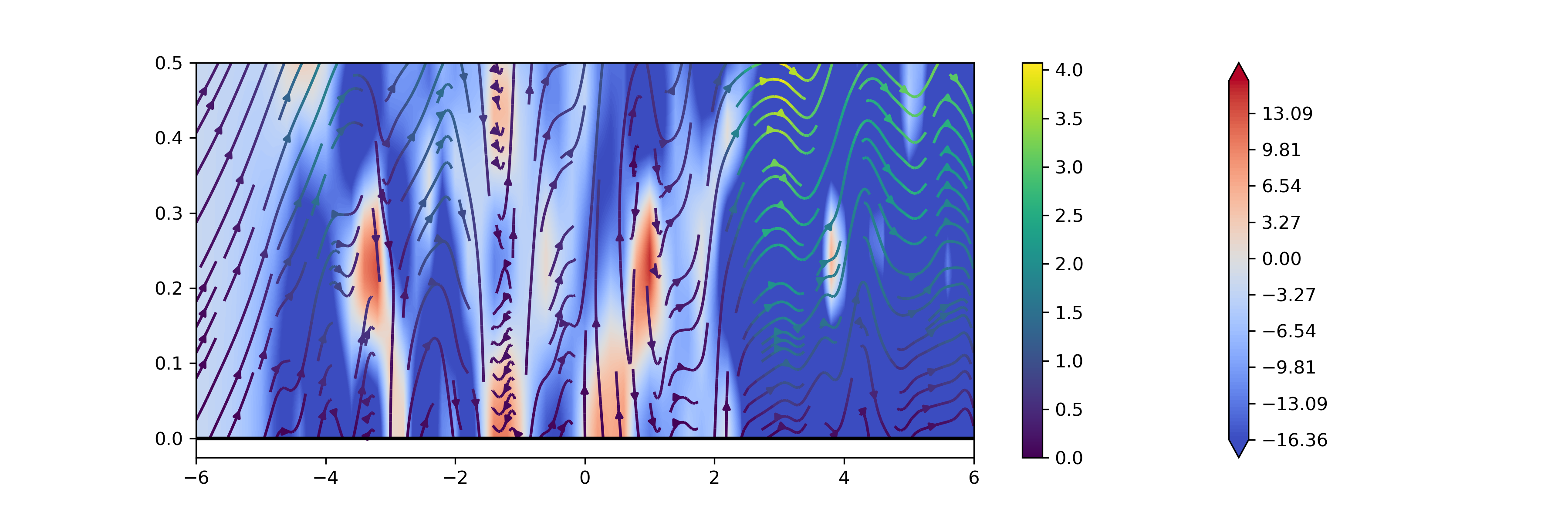}}}
    \caption{The boundary layer flow at different times $t$.}
    \label{Exp2FigBFlow}
\end{figure}

\section*{Data Availability Statement}

The data that support the findings of this study are available from
the corresponding author upon reasonable request.

\section*{Acknowledgement}

The authors would like to thank Oxford Suzhou Centre for Advanced
Research for providing the excellent computing facility. ZQ is supported
partially by the EPSRC Centre for Doctoral Training in Mathematics
of Random Systems: Analysis, Modelling and Simulation (EP/S023925/1).

\end{document}